\documentclass[draftclsnofoot,onecolumn,12pt,romanappendics]{IEEEtran}

\usepackage{amsfonts}
\usepackage{slashbox}
\usepackage{amsmath}
\usepackage{multirow}
\usepackage{graphicx}
\usepackage{amsfonts}
\usepackage{amsmath,epsfig}
\usepackage{mathbbold}
\usepackage{stfloats}
\usepackage{amssymb}
\usepackage{cite}
\usepackage{amsmath}
\usepackage{amssymb}
\usepackage{latexsym}
\usepackage{multirow}
\usepackage{epsfig}
\usepackage{graphics}
\usepackage{mathrsfs}
\usepackage{algorithmic}
\usepackage{algorithm}
\usepackage{subfigure}
\usepackage{slashbox}
\usepackage[all]{xy}

\usepackage{pifont}
\usepackage{bbding}
\usepackage{stmaryrd}
\usepackage{amssymb}
\usepackage{amsfonts}
\usepackage{epic}
\usepackage{stfloats}
\usepackage{latexsym}
\usepackage{epstopdf}
\usepackage{epic}
\usepackage{bm}
\usepackage{xcolor}
\usepackage{mathrsfs}
\usepackage{pifont}
\usepackage{bbding}
\usepackage{amsmath,epsfig}
\usepackage{mathbbold}
\usepackage{stmaryrd}
\usepackage{amssymb}
\usepackage{amsfonts}
\usepackage{epic}
\usepackage{graphicx}
\usepackage{subfigure}

\usepackage{enumerate}

\usepackage{stfloats}
\usepackage{latexsym}
\usepackage{epstopdf}
\usepackage{epic}
\usepackage{multirow}
\usepackage{stfloats}
\usepackage{bm}

\usepackage{booktabs}
\usepackage{color}
\newtheorem{theorem}{\underline{Theorem}}
\newtheorem{corollary}{\underline{Corollary}}
\newtheorem{lemma}{\underline{Lemma}}
\newtheorem{remark}{\underline{Remark}}
\begin{document}
\title{Energy-Efficient Small Cell with Spectrum-Power Trading}
\author{\IEEEauthorblockN{Qingqing Wu, \emph{Student Member, IEEE},  Geoffrey Ye Li, \emph{Fellow, IEEE},  \IEEEauthorblockN{Wen Chen},  \emph{Senior Member, IEEE}, and \IEEEauthorblockN{Derrick Wing Kwan Ng},  \emph{Member, IEEE}
\thanks{ Qingqing Wu and Geoffrey Ye Li are with the School of Electrical and Computer Engineering, Georgia Institute of Technology, USA, email:\{qingqing.wu, liye\}@ece.gatech.edu. Wen Chen  is with Department of Electronic Engineering, Shanghai Jiao Tong University, China, email: wenchen@sjtu.edu.cn. Derrick Wing Kwan Ng is with the School of Electrical Engineering and Telecommunications, The University of New South Wales, Australia, email: wingn@ece.ubc.ca. This paper has been submitted in part to the IEEE Globecom 2016.}}  }

\maketitle
\begin{abstract}
In this paper, we investigate spectrum-power trading between a small cell (SC) and a macro cell (MC), where the SC consumes power to serve the macro cell users (MUs) in exchange for some bandwidth from the MC.
Our goal is to maximize the system energy efficiency (EE) of the SC while guaranteeing the quality of service of each MU as well as small cell users (SUs). Specifically, given the minimum data rate requirement and the bandwidth provided by the MC, the SC jointly optimizes MU selection, bandwidth allocation, and power allocation while guaranteeing its own minimum required system data rate. The problem is challenging due to the binary MU selection variables and the fractional-form objective function. We first show that the bandwidth of an MU is shared with at most one SU in the SC. Then, for a given MU selection, the optimal bandwidth and power allocation is obtained by exploiting the fractional programming. To perform MU selection, we first introduce the concept of the trading EE to characterize the data rate obtained as well as the power consumed for serving an MU.  We then reveal a sufficient and necessary condition for serving an MU without considering the total power constraint and the minimum data rate constraint: the trading EE of the MU should be higher than the system EE of the SC. Based on this insight, we propose a low complexity MU selection method and also investigate the optimality condition.
 Simulation results verify our theoretical findings and demonstrate that the proposed resource allocation achieves  near-optimal performance.

\end{abstract}

\section{Introduction}
The fifth generation (5G) mobile networks are expected to provide ubiquitous ultra-high data rate services and seamless user experience across the whole system \cite{hu2014energy}. The concept of small cell (SC) networks, such as femtocells, has been recognized as a key technology that can significantly enhance the performance of 5G networks. The underlaying SCs enable the macro cells (MCs) to offload huge volume of data and large numbers of users \cite{han2015traffic}. In particular, the SC could help to serve some macro cell users (MUs) with high data rate requirements, especially when these MUs are far away from the MC base station (BS) \cite{shakir2013green}.
 Although the MUs offloading reduces the power consumption of MCs, additional power consumption is imposed to SCs that may degrade the quality of services (QoSs) of small cell users (SUs).  Therefore, motivating the SC  to serve MUs  is a critical problem,  especially when the SC BS does not belong to the same mobile operator with the MC BS \cite{han2013greening}.

  Meanwhile, the explosive growth of data hungry applications and various services has triggered a dramatic increase in energy consumption of wireless communications. Due to rapidly rising energy costs and tremendous carbon footprints \cite{ender15,niu2010cell,wu2012green,li2011energy,cui2004energy,wangxin2013}, energy efficiency (EE), measured in bits-per-joule, has attracted considerable attention as a new performance metric in both academia and industry \cite{qing15_wpcn_twc,li2014energy,ng2012energy1,ng2012energy3,ng2012energy2,cui2014optimal,sun2013energy,liu15,
  yiran2015,ramamonjison2015energy,han2014spectrum,sheenergy,ge2014energy}.
 Energy-efficient resource allocation has been studied in \cite{ng2012energy1} for a single cell with large numbers of base station antennas. Then, this work is extended into the context of physical layer security \cite{ng2012energy3}  and the multi-cell with  limited backhaul capacity  \cite{ng2012energy2}, respectively.
Subsequently, similar EE maximization problems are further investigated for example for relay \cite{cui2014optimal,sun2013energy},   full duplex \cite{liu15}, heterogenous \cite{yiran2015}, cognitive radio (CR) \cite{ramamonjison2015energy}, coordinated multi-point
(CoMP) transmission \cite{han2014spectrum}, and multi-input-multi-output orthogonal frequency
division multiplexing (MIMO-OFDM) \cite{sheenergy,ge2014energy} networks.
Furthermore, the authors in \cite{soh2013energy} propose a BS switching on-off scheme for heterogeneous cellular networks under a stochastic geometry model. It has been shown that significant power consumption is reduced after adopting strategic sleeping.

However, all these previous works ignored spectrum sharing or energy cooperation between the SC and the MC, which are expected to enhance the performances of both networks simultaneously. The notion of spectrum or energy cooperation has been recently pursued in \cite{xie2012energy,guo2014joint,cao2015cognitive,han2014enabling}.
 In \cite{xie2012energy},  energy-efficient resource allocation has been studied for heterogeneous cognitive radio networks with femtocells, where a cognitive BS  maximizes its profit by allocating the spectrum resource bought from the primary networks to the femtocells. However, only the spectrum disparity is exploited between these two communication networks. In \cite{guo2014joint}, joint energy and spectrum cooperation between two neighbouring cellular networks are considered to minimize the total costs on the pre-priced bandwidth and power given the QoS requirement. However, the monetary based spectrum sharing and energy cooperation are unable to capture the instantaneous characteristics of wireless channels \cite{cao2015cognitive,han2014enabling}.

  In this paper,  we study spectrum-power trading between an SC and an MC where the SC BS consumes additional power to serve MUs while the MC allows the SC BS to obtain additional bandwidth.   Specifically, the SC  BS splits the allocated bandwidth of an MU  into two parts. One part is allocated to meet the QoS of the MU and the other part can be utilized to serve its own SUs.
The spectrum-power trading is motivated by the following two observations. To serve the MUs that are far away from the MC BS,  the transmit power consumption limits the system performance rather than the bandwidth, since the MC BS generally operates in low signal-to-noise ratio (SNR) regimes \cite{tse2005fundamentals}. In contrast, for the SC, the bandwidth limits the system performance rather than the transmit power, since it generally works in high SNR regimes due to a small coverage \cite{tse2005fundamentals}. Thus, the spectrum-power trading in this paper will exploit the disparities between the MC networks and the SC networks from both the spectrum and the power perspectives.  Note that the spectrum-power trading is always beneficial to the MC by reducing its power consumption. Hence, we focus on how to enhance the performance of the SC.
   Although the spectrum-power trading  enables the SC to have higher data rate via seeking more bandwidth from the MUs, it also causes additional power consumption to the SC in order to serve the MUs. Thus, this may leave less transmit power for the SUs such that the SC operates in the low SNR regime. As a result, the power consumption becomes a critical problem.  In order to balance the power consumption and the achievable data rate, we adopt the system EE as the performance metric.

To ensure spectrum-power trading based EE, we need to address the following fundamental issues. First, when should the SC serve an MU?  For example,  if the required data rate of an MU is required too stringent but the bandwidth assigned to it is insufficient, it may not be beneficial for the SC to serve that MU.
 Second, how much bandwidth should be obtained and how much power should be utilized in order to achieve the maximum EE as well as guaranteeing the QoS of the MUs? This question arises because if the SC desires to seek more bandwidth from the MU, it has to transmit with a higher transmit power for this MU.
 However, this may in turn leave a lower transmit power for its own SUs and thereby lead to a lower system date rate as well as a  lower system EE. Thus, there exists a non-trivial spectrum and power tradeoff in the spectrum-power trading.
 These issues have never been investigated in previous works   \cite{qing15_wpcn_twc,li2014energy,ng2012energy1,ng2012energy2,ng2012energy3,cui2014optimal,liu15,sun2013energy,han2014spectrum,sheenergy,
yiran2015,ramamonjison2015energy,ge2014energy} and we will address them in this paper.
The main contributions are summarized as follows:
\begin{itemize}
\item We study spectrum-power trading between an SC and an MC where the SC consumes additional power to serve MUs in exchange for additional bandwidth from the MC. We focus on enhancing the performance of the SC. In particular, MU selection, bandwidth allocation, and power allocation are jointly optimized with the objective of  maximizing the system EE of the SC while guaranteeing the QoS of both networks.

\item  We first simplify the original optimization problem by showing that the bandwidth of an MU served by the SC is only shared with at most one SU. However, the simplified problem is still non-convex due to the binary MU selection variable and the fractional-form objective function. Given an MU selection, the problem can be further reduced to a joint bandwidth and power allocation problem, where the fractional form objective function is then transformed into a subtractive form by exploiting factional programming theory. We then derive closed-form expressions of the bandwidth and power allocation based on the analysis of the transformed problem.

  \item For the MU selection, we first introduce the trading EE of an MU that involves both the data rate brought for the SC and the power consumed by the SC in the spectrum-power trading. Then, we investigate the relationships between the trading EE of an MU and the system EE of the SC subject to various constraints in the original problem formulation.  In particular, we reveal that in the absence of the maximum power constraint and the minimum system data rate constraint, serving an MU can improve the system EE of the SC if and only if its trading EE is higher than the current system EE.
      Based on this observation, we develop a low computational complexity algorithm for the MU selection. Finally, we also study the optimality condition of the proposed algorithm for the original problem.
 \end{itemize}

The rest of this paper is organized as follows. Section II  introduces the spectrum-power trading model as well as the power consumption model. In Section III, we formulate and analyze the EE maximization problem. In Section IV, we study joint bandwidth and power allocation with a given MU selection. In Section V, we investigate the MU selection based on the proposed trading EE. Section VI provides simulation results to demonstrate the effectiveness of the proposed algorithm. Finally, Section VII concludes the paper.


\section{System Model}
\begin{figure}[t]
\centering
\includegraphics[width=4in]{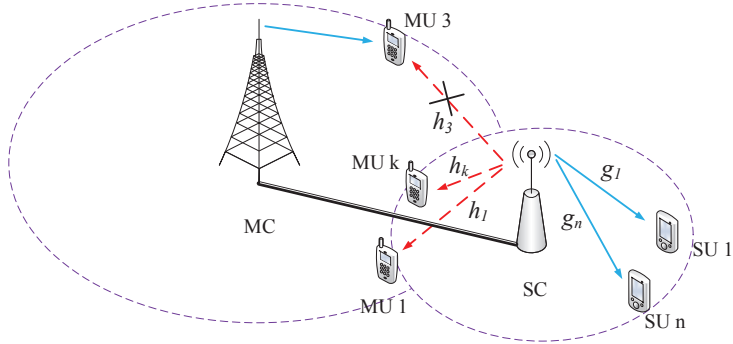}
\caption{The spectrum-power trading model between an SC and an MC.  For example, the SC may agree to serve MU 1 but refuse to serve MU 3 in order to maximize its performance. }\label{system_model}
\end{figure}
In this section, we first introduce the spectrum-power trading model between SC and MC networks. Then, we discuss the power consumption model of the SC BS under the context of spectrum-power trading. 
\subsection{Spectrum-Power Trading Model between SC and MC}

 We consider a spectrum-power trading scenario which consists of an MC and an SC, as depicted in Figure \ref{system_model}. The MC BS aims at offloading the data traffic of some cell edge MUs to the SC BS in order to reduce its own power consumption.  The set of MUs who may be served by the SC is denoted by $\mathcal{K}$ with $|\mathcal{K}|=K$ and the set of  SUs in the SC is denoted by $\mathcal{N}$ with $|\mathcal{N}|=N$, where $|\cdot|$ indicates the cardinality of a set. Each MU and SU have been assigned a licensed bandwidth by the MC and the SC, respectively, denoted as $W^k_{MC}$ and $B^n_{SC}$. To incentivize the SC to serve MUs, the MC allows it to utilize some of the licensed bandwidth of MUs to enhance the QoS of SUs. Thus, for the SC,  the bandwidth obtained from MUs can be viewed as a compensation of the power consumed for serving MUs.
To simplify the problem,  we assume that the SC BS as well  as each user is equipped with a single antenna \cite{guo2014joint}.


The channels between the SC and MUs as well as SUs are assumed to be quasi-static block fading, i.e., the channel coefficients remain constant during each block, but may vary from one block to another \cite{yumulti15}. We also assume that SU $n$, $\forall\, n\in \mathcal{N}$, experiences frequency flat fading on its own licensed bandwidth $B^n_{SC}$ and each MU $k$'s bandwidth $W^k_{MC}$, respectively. In addition,  MU $k$, $\forall\, k\in \mathcal{K}$, also experiences frequency flat fading on its own licensed bandwidth $W^k_{MC}$. Note that the results in this paper can also be extended to the more general case when the bandwidth of each user (SU and MU) is modeled by the multiple orthogonal subcarriers.
It is also assumed that the channel state information (CSI) of all users is perfectly known to the SC in order to explore the EE upper bound and extract useful design insights of the considered systems. In practice, the CSI can be estimated by each individual user and then fed back to the SC. Signaling overhead and imperfect CSI would result in performance loss and their impacts can be analyzed as in  [25], which is beyond the scope of this paper.


For MU $k$, $\forall\, k\in \mathcal{K}$, the channel power gain between the SC  and MU $k$ on its own licensed bandwidth $W_{MC}^k$ is denoted as $h_k$, cf. Figure \ref{system_model}. The corresponding transmit power and the bandwidth that are allocated to MU $k$ by the SC are denoted as $q_{k}$ and $w_k$, respectively.
Thus, the achievable data rate of MU $k$ can be expressed as
\begin{align}
r_k=w_k\log_2\left(1+\frac{q_kh_k}{w_kN_0}\right),
\end{align}
where $N_0$ is the spectral density of the additive white Gaussian noise.

For SU $n$, $\forall\, n\in \mathcal{N}$, the channel power gain between the SC  and SU $n$ on its own licensed bandwidth $B^n_{SC}$ is denoted as $g_{n}$, cf. Figure \ref{system_model}. The corresponding transmit power is denoted as $p_n$. Then, the achievable date rate of SU $n$ on its own bandwidth can be expressed as
\begin{align}
r^n_{SC}=B^n_{SC}\log_2\left(1+\frac{p_{n}g_{n}}{B^n_{SC}N_0}\right).
\end{align}
In addition to $B^n_{SC}$, each SU may obtain some additional bandwidth from MUs due to the spectrum-power trading between the SC and the MC.
Denote the channel power gain between the SC and SU $n$ on the bandwidth of MU $k$ as $g_{k,n}$. The bandwidth that the SC allocates for SU $n$ from $W^k_{MC}$ is denoted as as $b_{k,n}$ and the corresponding transmit power is denoted as $p_{k,n}$.
 Then, the achievable data rate of SU $n$ on the bandwidth of MU $k$ can be expressed as
\begin{align}
r_{k,n} = b_{k,n} \log_2\left(1+\frac{p_{k,n}g_{k,n}}{b_{k,n}N_0}\right).
\end{align}
Thus, the total data rate of SU $n$ in the context of the spectrum-power trading is given by
\begin{align}
R_n=B^n_{SC}\log_2\left(1+\frac{p_{n}g_{n}}{B^n_{SC}N_0}\right) + \sum_{k=1}^Kx_kb_{k,n} \log_2\left(1+\frac{p_{k,n}g_{k,n}}{b_{k,n}N_0}\right),
\end{align}
where $x_k$ is the MU selection variable and defined as
\begin{equation}\label{eq5}
x_k=\left\{
\begin{array}{lcl}
1,& \text{if MU $k$ is served by the SC},\\
0,& \text{otherwise}.
\end{array}\right.
\end{equation}
Therefore, the overall system data rate of SUs can be expressed as
\begin{align}\label{eq6}
R_{\rm{total}} =  \sum_{n=1}^{N}R_n = \sum_{n=1}^Nr^n_{SC} + \sum_{n=1}^N\sum_{k=1}^Kx_kr_{k,n}.
\end{align}

\subsection{Power Consumption Model for SC BS}
 Here, we adopt the power consumption model from \cite{ismailsurvey14} in which the overall energy consumption of the BS consists of two parts: the dynamic power consumed in the power amplifier for transmission, $P_{\rm{t}}$, and the static power consumed for circuits, $P_{\rm{c}}$.

  The dynamic power consumption is modeled as a linear function of the transmit power that includes both the transmit power consumption for SUs and that for MUs, i.e.,
 \begin{align}\label{eq7}
P_{\rm{t}}=\sum_{n=1}^{N}\frac{p_{n}}{\xi} + \sum_{n=1}^{N}\sum_{k=1}^{K}x_k\frac{p_{k,n}}{\xi}+\sum_{k=1}^{K}x_k\frac{q_{k}}{\xi},
\end{align}
where $\xi\in (0,1]$  is a constant that accounts for the power amplifier (PA) efficiency and the value of  $\xi$ depends on the specific type of the BS. In general, the PA efficiency decreases for smaller BS types and a detailed discussion on it can be found in \cite{auer2010d2}.
The static power consumption for circuits is denoted as  $P_{\rm{c}}$, which is caused by filters, frequency synthesizer, etc.  Therefore, the overall power consumption of the SC BS can be expressed as
\begin{align}\label{eq7}
P_{\rm{total}}=\sum_{n=1}^{N}\frac{p_{n}}{\xi} + \sum_{n=1}^{N}\sum_{k=1}^{K}x_k\frac{p_{k,n}}{\xi}+\sum_{k=1}^{K}x_k\frac{q_{k}}{\xi}+P_{\rm{c}}.
\end{align}

\section{Problem Formulation and Analysis}
    Our goal is to enhance the system EE of the SC in the context of the spectrum-power trading while guaranteing the QoS of the MUs as well as the SC network. Thus, the system EE of the SC is defined as the ratio of the total achievable data rate of SUs and the total power consumption that includes not only the power consumed for provding services for SUs, but also the power consumed for spectrum-power trading, i.e.,
\begin{align}\label{eq5}
EE=\frac{R_{\rm{total}}}{P_{\rm{total}}}=\frac{\sum_{n=1}^Nr^n_{SC} + \sum_{k=1}^K\sum_{n=1}^Nx_kr_{k,n}}
{\sum_{n=1}^{N}\frac{p_{n}}{\xi} + \sum_{n=1}^{N}\sum_{k=1}^{K}x_k\frac{p_{k,n}}{\xi}+\sum_{k=1}^{K}x_k\frac{q_{k}}{\xi}+P_{\rm{c}}}.
\end{align}
Specifically, we aim to maximize the system EE of the SC via jointly optimizing MU selection, bandwidth allocation, and power allocation. The system EE maximization problem is formulated as
\begin{align}\label{eq16}
\mathop {\text{maximize} }\limits_{\overset{\{p_{n}\}, \{p_{k,n}\},\{b_{k,n}\},}{\{x_{k}\}, \{q_k\}, \{w_k\}}} ~~~& \frac{\sum_{n=1}^Nr^n_{SC} + \sum_{k=1}^K\sum_{n=1}^Nx_kr_{k,n}}
{\sum_{n=1}^{N}\frac{p_{n}}{\xi} + \sum_{n=1}^{N}\sum_{k=1}^{K}x_k\frac{p_{k,n}}{\xi}+\sum_{k=1}^{K}x_k\frac{q_{k}}{\xi}+P_{\rm{c}}}\nonumber \\
\text{s.t.} ~~~~~~~~&  \text{C1:}~~\sum_{n=1}^{N}{p_{n}} +  \sum_{n=1}^N\sum_{k=1}^Kp_{k,n}+\sum_{k=1}^{K}q_k\leq P^{SC}_{\mathop{\max}}, \nonumber \\
& \text{C2:}~~\sum_{n=1}^Nb_{k,n}+w_k\leq x_kW^k_{MC},  ~\forall\, k\in\mathcal{ K}, \nonumber \\
&\text{C3:}~~w_{k}\log_2\left(1+\frac{q_{k}h_{k}}{w_{k}N_0}\right)\geq x_kR^k_{MC}, ~\forall\, k\in\mathcal{ K},  \nonumber \\
& \text{C4:}~~\sum_{n=1}^Nr^n_{SC} + \sum_{k=1}^K\sum_{n=1}^Nx_kr_{k,n} \geq R^{SC}_{\mathop{\min}}, \nonumber\\
& \text{C5:}~~x_k\in\{0, 1\}, ~\forall\,  k\in\mathcal{ K},    \nonumber\\
& \text{C6:}~~b_{k,n}\geq  0, ~  w_{k}\geq  0, ~\forall\, k\in\mathcal{ K}, n\in\mathcal{N},    \nonumber\\
& \text{C7:}~~p_n\geq0,~ p_{k,n}\geq0, ~  q_{k}\geq0, ~\forall\,  k\in\mathcal{ K}, n\in\mathcal{N}.
\end{align}
In problem (\ref{eq16}),  C1 limits the maximum transmit power of the SC BS to $P^{SC}_{\mathop{\max}}$.
C2 ensures that the bandwidth allocated to SUs and MU $k$ does not exceed the available bandwidth, $W^k_{MC}$, that has been licensed to MU $k$ by the MC.
In C3, $R^k_{MC}$ is the minimum data rate requirement of MU $k$.  C4 guarantees the minimum required system data rate of the SC.
C5 indicates whether to serve MU $k$ or not. Note that if $x_k=0$, then from C2 and C4, both $b_{k,n}$ and $q_k$ will be forced to be zeros at the optimal solution of problem  (\ref{eq16}), which means that the SC does not obtain additional bandwidth from MU $k$ and does not serve MU $k$ either.
 C6 and C7 are non-negativity constraints on the bandwidth and power allocation variables, respectively.  In general, different priorities and fairness among the SUs could be realized by adopting the weighted sum rate instead of the sum rate in problem (\ref{eq16}) \cite{ng2012energy1,derrick_harvest13}.  Since the weights do not affect the algorithm design, we assume that all the SUs are equally weighted in this paper without loss of generality.
%
{\begin{remark}
Although we focus on improving the EE of the SC via spectrum-power trading, the EE of the MC as well as the system-wide EE will also be improved correspondingly, which can be explained as follows.
 Note that the MUs that the MC is willing to offload to the SC are in general those users with poor channel conditions or cell-edge users. This means that a large amount of transmit power will be consumed if these MUs are directly served by the MC. In fact, this is also the fundamental reason why the MC desires to offload them.  That is, through offloading, the MC only needs to serve MUs with good channel conditions, which thus results in a higher system EE. In other words, if an MU can be served by the MC with a small amount of transmit power, there is no motivation for the MC to establish the offloading. Therefore, the EE of the MC will obviously increase via offloading.  Therefore, the system-wide EE will increase due to the EE increases of both the SC and the MC.
\end{remark}}

\begin{remark}
It is worth noting that problem (\ref{eq16}) generalizes several interesting special cases which are discussed as follows.
\begin{itemize}
  \item If we set $x_k=0, \forall\, k\in \mathcal{K}$, then problem (\ref{eq16}) is reduced to a system EE maximization probem without spectrum-power trading.
  \item If we set $x_k=1,\, \forall\, k\in \mathcal{K}$, it suggests that the SC helps to provide services for all of the MUs without considering its own performance, which usually happens when the SC BS and the MC BS belongs to the same operator.
  \item If we set $B^n_{SC}=0, \forall\, n\in \mathcal{N}$, it implies that the SC does not have its own licensed bandwidth to assign to SUs and can only seek the bandwidth from the MC via the spectrum-power trading.
In this case, the SC is reduced to \emph{a cognitive (secondary) network} while the MC can be regarded as \emph{a primary network}  \cite{cao2015cognitive}.
\end{itemize}
Therefore,  problem (\ref{eq16}) is more challenging and  more interesting than the previous work \cite{han2014enabling}.
\end{remark}

Note that problem (\ref{eq16}) is neither a concave nor  a quasi-concave optimization problem due to the fractional-form objective function and the binary optimization variables $x_k, \forall\, k$. 
Nevertheless, in the following theorem proved in Appendix A, we first transform the energy-efficient optimization problem into a simplified one based on its special structure.
\begin{theorem}
The optimal solution of problem (\ref{eq16}) is equivalent to that of the following problem
\begin{align}\label{eq17}
\mathop {\text{maximize} }\limits_{\overset{\{p_{n}\}, \{p_{k,k'}\},\{b_{k,k'}\},}{\{x_{k}\}, \{q_k\}, \{w_k\}}} ~~~& \frac{\sum_{n=1}^Nr^n_{SC} + \sum_{k=1}^Kx_kr_{k,k'}}
{\sum_{n=1}^{N}\frac{p_{n}}{\xi} + \sum_{k=1}^{K}x_k\frac{p_{k,k'}}{\xi}+\sum_{k=1}^{K}x_k\frac{q_{k}}{\xi}+P_{\rm{c}}}\nonumber \\
\text{s.t.} ~~~~~~~~&  \text{C5}, ~ \text{C6},~ \text{C7},~  \nonumber \\
& \text{C1:}~~ \sum_{n=1}^{N}{p_{n}} + \sum_{k=1}^Kp_{k,k'}+\sum_{k=1}^{K}q_k\leq P^{SC}_{\mathop{\max}}, \nonumber \\
& \text{C2:}~~b_{k,k'}+w_k= x_kW^k_{MC},  ~\forall\, k\in\mathcal{ K}, \nonumber \\
&\text{C3:}~~w_{k}\log_2\left(1+\frac{q_{k}h_{k}}{w_{k}N_0}\right)= x_kR^k_{MC}, ~\forall\, k\in\mathcal{ K},  \nonumber \\
& \text{C4:}~~\sum_{n=1}^Nr^n_{SC} + \sum_{k=1}^Kx_kr_{k,k'} \geq R^{SC}_{\mathop{\min}}, 
\end{align}
where $k'=\arg\mathop {\max }\limits_{n \in \mathcal{N}} ~g_{k,n}.$
\end{theorem}

Theorem 1 suggests that if the SC decides to serve MU $k$, the most energy-efficient strategy is only to share the bandwidth of MU $k$ with one SU who has the largest channel power gain on the traded bandwidth, $W^k_{MC}$. In addition, constraints C2 and C3 are also met with equalities at the optimal solution since it is always beneficial for the SC to seek as much as bandwidth while consuming as less as transmit power in the spectrum-power trading with the MC.
With Theorem 1, we only need to focus on solving problem  (\ref{eq17}) in the rest of the paper.
Although problem (\ref{eq17}) is more tractable than problem  (\ref{eq16}), it is still a combinatorial non-convex optimization problem. In general, there is no efficient method for this problem and the exhaustive search among all the possible cases leads to an exponential computational complexity, which is prohibitive in practice. Thus,  we aim to develope a low complexity approach via exploiting the special structure of the problem.

\section{Energy-Efficient Resource Allocation For Given MU Selection}
Denote $\Psi$ as a set of MUs that are scheduled by the SC, i.e., $\Psi\triangleq\{k\,|\,x_k=1, k\in \mathcal{K}\}$, and denote $EE_{\Psi}$ as the maximum system EE of problem (\ref{eq17}) based on set $\Psi$, i.e., $EE=EE_{\Psi}$. For a given $\Psi$, problem (\ref{eq17}) is reduced to a joint bandwidth and power allocation problem.
However, the reduced problem is still non-convex due to the fractional-form objective function.
In the following, we show that the optimal solution of the reduced problem can be efficiently obtained by exploiting the fractional structure of the objective function in (\ref{eq17}).
\subsection{Problem Transformation}
According to the nonlinear fractional programming theory \cite{dinkelbach1967nonlinear}, for a problem of the form,
\begin{equation}\label{eq11c}
q^*= \mathop {\text{maximize} }\limits_{S\in \mathcal{F}} \frac{R_{\rm{\rm{total}}}(S)}{P_{\rm{\rm{total}}}(S)},
\end{equation}
where $S$ is a feasible solution and $\mathcal{F}$ is the corresponding feasible set,
there exists an equivalent problem in  subtractive form that satisfies
\begin{equation}\label{eq12c}
T(q^*)=\mathop {\text{maximize} } \limits_{S\in \mathcal{F}}\Big\{R_{\rm{\rm{total}}}(S)-q^*P_{\rm{\rm{total}}}(S)\Big\}=0.
\end{equation}
The equivalence between (\ref{eq11c}) and (\ref{eq12c}) can be easily verified with the corresponding maximum value $q^*$ that is also the maximum system EE. Besides, Dinkelbach
 provides an iterative  method in \cite{dinkelbach1967nonlinear} to obtain  $q^*$. Specifically,  for a given $q$, we solve a maximization problem with the subtractive-form objective function as (\ref{eq12c}). The value of $q$ is then updated and problem (\ref{eq12c}) is solved again in the next iteration until convergence.
  By applying this transformation to (\ref{eq17}) with  $b_{k,k'} = W^k_{MC}- w_k$ and $q_{k}=\left(2^\frac{R^{k}_{MC}}{w_k}-1\right)\frac{w_{k}N_0}{h_k}, \forall \, k \in \Psi$, we obtain the following optimization problem for a given $q$ in each iteration
\begin{align}\label{eq13c1}
\mathop {\text{maximize} }\limits_{\{p_{n}\}, \{p_{k,k'}\}, \{w_k\}}~~~ &\sum_{n=1}^NB^n_{SC}\log_2\left(1+\frac{p_{n}g_{n}}{B^n_{SC}N_0}\right) + \sum_{k\in\Psi}(W^k_{MC}-w_k)\log_2\left(1+\frac{p_{k,k'}g_{k,k'}}{(W^k_{\max}-w_k)N_0}\right) \nonumber\\
&-q\left(\sum_{n=1}^{N}\frac{p_{n}}{\xi} + \sum_{k\in \Psi}\frac{p_{k,k'}}{\xi}+\sum_{k\in \Psi}\left(2^\frac{R^{k}_{MC}}{w_k}-1\right)\frac{w_{k}N_0}{\xi h_k}+P_{\rm{c}}\right)\nonumber\\
\text{s.t.} ~~~~~~~~& \text{C1:}~~ \sum_{n=1}^{N}{p_{n}} + \sum_{k\in \Psi}p_{k,k'}+\sum_{k\in \Psi}\left(2^\frac{R^{k}_{MC}}{w_k}-1\right)\frac{w_{k}N_0}{h_k}\leq P^{SC}_{\mathop{\max}}, \nonumber \\
& \text{C4:}~~\sum_{n=1}^Nr^n_{SC} + \sum_{k\in \Psi}r_{k,k'} \geq R^{SC}_{\mathop{\min}}, ~~~ {\text{C6:}~~W^{k}_{MC} \geq w_{k}\geq  0, ~\forall\, k\in \Psi,}    \nonumber\\ 
& \text{C7:}~~p_n\geq0, ~p_{k,k'}\geq0, ~\forall\,  k\in\Psi, n\in\mathcal{N}.
\end{align}
After the transformation, it can be verified that problem (\ref{eq13c1}) is jointly concave with respect to all optimization variables and also satisfies Slater's constraint qualification \cite{Boyd}.  As a result, the duality gap between problem (\ref{eq13c1}) and its dual problem is zero, which means that the optimal solution of problem (\ref{eq13c1}) can be obtained by applying the Lagrange duality theory \cite{ng2012energy3}. In the next section, we will derive the optimal bandwidth and power allocation via exploiting the Karush-Kuhn-Tucker (KKT) conditions of problem (\ref{eq13c1}) that leads to a computationally efficient algorithm.

\subsection{Joint Bandwidth and Power Allocation}
The partial  Lagrangian function of  problem (\ref{eq13c1}) can be written as
\begin{align}\label{eq161}
~\,~&~\mathcal{L}(p_{n}, p_{k,k'}, w_k, \lambda, \mu)  \nonumber \\
=\,&\sum_{n=1}^NB^n_{SC}\log_2\left(1+\frac{p_{n}g_{n}}{B^n_{SC}N_0}\right) + \sum_{k\in\Psi}(W^k_{MC}-w_k)\log_2\left(1+\frac{p_{k,k'}g_{k,k'}}{(W^k_{MC}-w_k)N_0}\right) \nonumber\\
-\,&q\left(\sum_{n=1}^{N}\frac{p_{n}}{\xi} + \sum_{k\in \Psi}\frac{p_{k,k'}}{\xi}+\sum_{k\in \Psi}\left(2^\frac{R^{k}_{MC}}{w_k}-1\right)\frac{w_{k}N_0}{ \xi h_k}+P_{\rm{c}}\right)\nonumber\\
+\,& \lambda\left(P^{SC}_{\mathop{\max}} -\sum_{n=1}^{N}{p_{n}} - \sum_{k\in \Psi}{p_{k,k'}}-\sum_{k\in \Psi}\left(2^\frac{R^{k}_{MC}}{w_k}-1\right)\frac{w_{k}N_0}{ h_k}\right) \nonumber \\
+\,&   \mu\left(\sum_{n=1}^NB^n_{SC}\log_2\left(1+\frac{p_{n}g_{n}}{B^n_{SC}N_0}\right) + \sum_{k\in\Psi}(W^k_{MC}-w_k)\log_2\left(1+\frac{p_{k,k'}g_{k,k'}}{(W^k_{MC}-w_k)N_0}\right) -R^{SC}_{\mathop{\min}}\right),
\end{align}
where  $\lambda$ and  $  \mu$ are the non-negative Lagrange multipliers associated with  constraints C1 and C4, respectively.  The boundary constraints C6 and C7 are absorbed into the optimal solution in the following.
Then, from Appendix B, the optimal solution can be obtained as in Theorem \ref{theorem70}.
 \begin{theorem}\label{theorem70}
Given  $\lambda$ and ${\mu}$,  the optimal bandwidth and power allocation of maximizing the Lagrangian function, $\mathcal{L}$, is given by
\begin{align}
w_{k}&= \min\left(\frac{R^k_{MC}\ln2}{\mathcal{W}\left(\frac{1}{e}\left(\frac{\mathcal{C}h_k}{(q+\lambda)N_0}-1\right)\right)+1}, W^k_{MC}\right), \forall \, k \in \Psi,  \label{eq4.22}\\
p_{k,k'}&=(W^k_{MC}-w_k)\left[\frac{(1+\mu)\xi}{({q}+\lambda\xi)\ln2}-\frac{N_0}{g_{k,k'}}\right]^+, \forall \, k \in \Psi,  \label{eq4.25}\\
p_{n}&=B^n_{SC}\left[\frac{(1+\mu)\xi}{({q}+\lambda\xi)\ln2}-\frac{N_0}{g_{n}}\right]^+, \forall \, n\in \mathcal{N},  \label{eq4.26}
\end{align}
where $[x]^+\triangleq \max\{x,0\}$ and $\mathcal{W}(x)$ is the Lambert $\mathcal{W}$ function \cite{corless1996lambertw},  i.e., $x=\mathcal{W}(x)e^{\mathcal{W}(x)}$. In addition, $\mathcal{C}= (1+\mu)\log_2\left(1+{\widetilde{p}}_{k,k'}\frac{{g_{k,k'}}}{N_0}\right)-\left(\frac{q}{\xi}+\lambda\right)\widetilde{p}_k$, and ${\widetilde{p}}_{k,k'}=\left[\frac{(1+\mu)\xi}{({q}+\lambda\xi)\ln2}-\frac{N_0}{g_{k,k'}}\right]^+$.

 \end{theorem}

From (\ref{eq4.22}),  it is easy to show that the bandwidth allocated to MU $k$ by the SC, i.e., $w_k$, increases with its minimum required data rate by the MC, $R^{k}_{MC}$, while decreasing with its channel power gain, $h_k$.  This implies that the SC is able to seek more bandwidth from the MUs who require lower user data rates but are closer to the SC BS, which also coincides with the intuition discussed previously. {Furthermore,
we also observe that the optimal transmit power allocations, $p_{k,k'}$ and $p_n$, follow the conventional multi-level water-filling structure due to different bandwidth allocations. In contrast, the optimal transmit power densities, $\frac{p_{k,k'}}{W^k_{MC}-w_k}$ and $\frac{p_n}{B^n_{SC}}$, follow the single-level water-filling structure \cite{ng2012energy3}.} Given Lagrange multipliers $\lambda$ and $\mu$, the optimal bandwidth and power allocation can be obtained immediately from Theorem \ref{theorem70}.

{{\begin{table}[]
\caption{\small{Energy-Efficient Joint Bandwidth and Power Allocation}} \label{tab:overall} \centering
\vspace{-0.8cm}
\begin{algorithm}[H]
 \caption{Energy-Efficient Joint Bandwidth and Power Allocation Algorithm} 
  \begin{algorithmic}[1]
\STATE \normalsize {  {\bf{Initialize}} the maximum accuracy  $\epsilon$ and set $q=1$ with given MU $k$ and SU $k'$;\\ 
\REPEAT 
\STATE       Initialize ${\lambda}$ and ${\mu}$; \\
\REPEAT
\STATE         Obtain  $w_k$, $p_{k,k'}$, and $p_n$ from (\ref{eq4.22})-(\ref{eq4.26});
\STATE         Update dual variables ${\lambda}$ and ${\mu}$ from (\ref{eq4.30}) and (\ref{eq4.34});
\UNTIL{${\lambda}$ and ${\mu}$ converge};
\STATE         Update $q$ from (\ref{eq11c});
\UNTIL{ $\Big(R_{\rm{\rm{tot}}}(S)-qP_{\rm{\rm{tot}}}(S)\Big) \leq \epsilon$. } }
\end{algorithmic}
\end{algorithm}
\end{table}}}

After computing the primal variables ($p_n, p_{k,k'}$, $w_{k}$), we now proceed to solve the dual problem, i.e., $\mathop{\text{minimize}} \limits_{{\lambda}\geq {0}, \mu\geq 0}~ \mathcal{G}({\lambda}, {\mu})$, where $\mathcal{G}({\lambda}, {\mu})=\mathop{\text{maximize}} \limits_{p_{n}, p_{k,k'}, w_k}~ \mathcal{L}(p_{n}, p_{k,k'}, w_k, \lambda, \mu)$.
Since a dual function is always convex by definition,  the commonly used ellipsoid method
 can be employed for updating $({\lambda}, {\mu})$ toward the optimal solution with guaranteed convergence \cite{Boyd}.
In addition, it has been pointed in \cite{yu2006dual} that
 the ellipsoid method is able to converge faster and more stable across a wide variety of situations.
 Thus, in this paper,  we adopt the ellipsoid method to update the Lagrange multipliers and the subgradients that will be used are given by
\begin{align}
\triangle\lambda&= P^{SC}_{\mathop{\max}} -\sum_{n=1}^{N}\frac{p_{n}}{\xi} - \sum_{k\in \Psi}\frac{p_{k,k'}}{\xi}-\sum_{k\in \Psi}\frac{q_{k}}{\xi}, \label{eq4.30}\\
\triangle\mu&= \sum_{n=1}^Nr^n_{SC} + \sum_{k\in \Psi}r_{k,k'} -  R^{SC}_{\mathop{\min}}. \label{eq4.34}
\end{align}
 A discussion regarding the choice of the initial ellipsoid,  the updating of the ellipsoid, and the stopping criterion for the ellipsoid method can be found in \cite{yu2006dual} (Section IV-B) and is thus omitted here for brevity.  The updated Lagrange multipliers in (\ref{eq4.30}) and (\ref{eq4.34}) can be used to obtain the bandwidth and power allocation variables in the primary variable optimization. Due to the concavity of primary problem (\ref{eq13c1}), the iterative optimization between $(p_n, p_{k,k'}, w_k)$ and $({\lambda}, {\mu})$ is guaranteed to converge to the optimal solution of (\ref{eq13c1}). The details of the bandwidth and power allocation for a given MU selection are summarized in Algorithm 1 in Table I.

\section{Energy-Efficient MU Selection}
In this section, we investigate the MU selection problem, i.e., to find the MU set $\Psi$ where $x_k=1,\forall\, k\in\Psi$. We first introduce the concept of the  trading EE that plays a key role in the algorithm development. Then, we study the MU selection condition under different cases and propose a low computational complexity algorithm based on the trading EE. Finally, we analyze the computational complexity of the proposed algorithm.
 \subsection{Trading EE}
  The Trading EE of MU $k$, $\forall\, k\in \mathcal{K}$, is defined as the total data rate of MU $k$ brought for the SC over the total power consumed by the SC in the spectrum-power trading, i.e.,
\begin{align}
EE_{k}=\frac{b_{k,k'}\log_2\left(1+\frac{p_{k,k'}g_{k,k'}}{b_{k,k'}N_0}\right)}{ \frac{p_{k,k'}}{\xi}+\frac{q_k}{\xi}},
\end{align}
where the numerator, $b_{k,k'}\log_2\left(1+\frac{p_{k,k'}g_{k,k'}}{b_{k,k'}N_0}\right)$, is the additional data rate that the SC obtains via serving MU $k$ and the denominator, $\frac{p_{k,k'}}{\xi}+\frac{q_k}{\xi}$, is the total power consumed for both supporting SU $k'$ and meeting the QoS of MU $k$.  As a result, the trading EE is in fact an evaluation of an MU in terms of the power utilization efficiency and can be regarded as a profit of the SC in the spectrum-power trading.
Then, the trading EE maximization problem of MU $k$ can be formulated as
 \begin{align}\label{eq20}
\mathop {\text{maximize} }\limits_{p_{k,k'}, b_{k,k'}, q_k, w_k} ~~~& EE_{k}=\frac{b_{k,k'}\log_2\left(1+\frac{p_{k,k'}g_{k,k'}}{b_{k,k'}N_0}\right)}{ \frac{p_{k,k'}}{\xi}+\frac{q_k}{\xi}}\nonumber \\
\text{s.t.} ~~~~~~~~
& \text{C2:}~~b_{k,k'}+w_k\leq W^k_{MC},  \nonumber \\
&\text{C3:}~~w_{k}\log_2\left(1+\frac{q_{k}h_{k}}{w_{k}N_0}\right)\geq R^k_{MC},  \nonumber \\
& \text{C7:}~~b_{k,k'}\geq  0, ~ w_{k}\geq  0.    
\end{align}
It is worth noting that problem (\ref{eq20}) can be regarded as a special case of problem (\ref{eq16}) where there is only one MU and one SU. Therefore, problem (\ref{eq20}) can be solved similarly by the algorithm proposed in Section III. However, in order to provide more insight, we show that the optimal solution can be solved more efficiently in the following theorem that is proved in Appendix C.
\begin{theorem}\label{trading}
Problem (\ref{eq20}) is equivalent to the following quasi-concave maximization problem
\begin{align}\label{eq201}
\mathop {\text{maximize} }\limits_{p_{k,k'}\geq0,\, w_k\geq0} ~~~& EE_{k}=\frac{(W^k_{MC}- w_k)\log_2\left(1+\frac{p_{k,k'}g_{k,k'}}{(W^k_{MC}- w_k)N_0}\right)}{ \frac{p_{k,k'}}{\xi}+\left(2^\frac{R^k_{MC}}{w_k}-1\right)\frac{w_kN_0}{h_k\xi}},
\end{align}
where $EE_k$ is strictly and jointly quasi-concave over $p_{k,k'}$ and $w_k$.
\end{theorem}

Since $EE_k$ is strictly and jointly quasi-concave over $p_{k,k'}$ and $w_k$ under a convex feasible set, the optimal solutions of $p_{k,k'}$ and $w_k$ are both unique. {This suggests that the alternating method, also known as coordinated descent method \cite{bertsekas1999nonlinear}, can be employed to obtain the optimal $p_{k,k'}$ and $w_k$ efficiently \cite{miao2010energy}.  Specifically, for a given $w_k$ or  $p_{k,k'}$, problem  (\ref{eq201}) is simplified  into a univariate quasi-concave maximization with respect to $p_{k,k'}$ or $w_k$,  where the optimal values can be easily obtained by the bisection method \cite{Boyd}. For example, for a given $w_k$, it has been shown in \cite{miao2010energy} that by judging the derivative of $EE_k$ with respect to $p_{k,k'}$ is zero or not, we can obtain the optimal $p_{k,k'}$.
}


 \subsection{Trading EE based MU Selection}
 The key observation of the user trading EE is that both $b_{k,k'}\log_2\left(1+\frac{p_{k,k'}g_{k,k'}}{b_{k,k'}N_0}\right)$ and $\frac{p_{k,k'}}{\xi}+\frac{q_k}{\xi}$ will be removed respectively from the numerator and the denominator of the objective function in problem (\ref{eq17}) if MU $k$ is not served by the SC.
 With the user trading EE defined in Section III-A, we now investigate the MU selection conditions for different cases. Recall that $\Psi$ denotes an arbitrary set of  MUs that are scheduled by the SC, i.e., $\Psi\triangleq\{k\,|\,x_k=1, k\in \mathcal{K}\}$, and  $EE^*_{\Psi}$ denotes the maximum system EE of problem (\ref{eq17}), which can be obtained by Algorithm 1 based on set $\Psi$. Then, we have the following theorem, proved in Appendix D, to facilitate the algorithm development.
 \begin{theorem}\label{scheduling}
 For any unscheduled MU $m$, i.e.,  $m\in \mathcal{K}, m\notin \Psi$:
 \begin{enumerate}
 \item in the absence of constraints C1 and C4 in problem (\ref{eq17}),  serving MU $m$ improves the EE of the SC \emph{if and only if} $EE^*_m> EE^*_{\Psi}$;
   \item in the absence of constraint C1  in problem (\ref{eq17}),  serving MU $m$ improves the EE of the SC \emph{if} $EE^*_m> EE^*_{\Psi}$;
   \item in the absence of constraint C4  in problem (\ref{eq17}),  serving MU $m$ improves the EE of the SC \emph{only if} $EE^*_m>EE^*_{\Psi}$.
 \end{enumerate}
\end{theorem}

Theorem \ref{scheduling} reveals the relationship between the inequality $EE^*_k> EE^*_{\Psi}$ and the MU selection under different constraints in problem  (\ref{eq17}).  Specifically, without considering both the maximum power constraint and the system minimum data rate constraint, $EE^*_k> EE^*_{\Psi}$ is the sufficient and necessary condition for serving MU $k$.  Besides, without considering the maximum power constraint, $EE^*_k> EE^*_{\Psi}$ is reduced to a sufficient condition for serving MU $k$. In contrast, without considering the minimum system  data rate constraint, $EE^*_k> EE^*_{\Psi}$ is reduced to a necessary condition for serving MU $k$.
 Since these two constraints, i.e., C1 and C4,  may not be met with equalities simultaneously in most cases, it means that $EE^*_k> EE^*_{\Psi}$ is either sufficient or necessary for serving MU $k$ in practice.  It is also interesting to mention that $EE^*_k> EE^*_{\Psi}$ has an important practical interpretation:  the trading EE  of MU $k$ that is selected by the SC should be higher than the current EE of the SC. In other words, the spectrum-power trading on this MU enables the SC to have a better utilization of the power. Otherwise, the spectrum-power trading is only beneficial to the MC and does not bring any benefit for the SC.

The main implication of Theorem \ref{scheduling} is that an MU with higher user trading EE is more likely to be scheduled by the SC.
Based on this insight, a computationally efficient MU selection scheme is designed as follows. First, sort all the MUs in the descending order according to the user trading EE. Second, for MU $k$ satisfying the condition  $EE^*_k> EE^*_{\Psi}$ in Theorem \ref{scheduling}, set $x_k=1$ and maximize the system EE in problem (\ref{eq17}) by Algorithm 1. Third, by comparing the updated system EE with previous system EE where $x_k=0$ holds, decide whether to schedule MU $k$.  The details of the MU selection procedure is summarized in Algorithm 2 in Table II. To understand Algorithm 2 better, we provide the following corollary to characterize the optimality condition that has been proved in Appendix E.
{\begin{corollary}\label{optmality}
Algorithm 2 is optimal for problem (\ref{eq17}) in the absence of constraints C1 and C4.
\end{corollary}}

Corollary 1 reveals that Algorithm 2 achieves the maximum system EE of the SC if constraints C1 and C4 are not considered. This can be interpreted as follows. Without considering C1 and C4,  Theorem \ref{scheduling} points out that an MU with trading EE higher than the current system EE is sufficient and necessary to be scheduled. In addition, the updated system EE after scheduling the MU is still lower than the trading EE of this MU. This indicates that if  MU $k$ is scheduled in the optimal solution, then the MUs with higher trading EE than MU $k$ should also be scheduled.

{{\begin{table}[]
\caption{\small{Energy-Efficient Spectrum-Power Trading Algorithm}} \label{tab:overall} \centering
\vspace{-0.8cm}
\begin{algorithm}[H]
 \caption{Energy-Efficient Spectrum-Power Trading Algorithm} 
  \begin{algorithmic}[1]
\STATE \normalsize {  Obtain $EE_k$, $\forall\, k$, by solving problem (\ref{eq201}) in Theorem 3; \\
\STATE Sort all MUs in the descending order of trading EE, i.e., $EE_{1}^{*} > EE_{2}^{*}>,...,> EE_{K}^{*}$; \\
\STATE Set  $\Psi={\O}$ and obtain $EE_{\Psi}^{*}$ by Algorithm 1; \\  \STATE {\bf{for}}  $k=1:K$\\
\STATE       ~~~ Obtain $EE^*_{\Psi \bigcup \{k\}}$ by Algorithm 1; \\
\STATE      ~~~ {\bf{if}}    $EE^*_{\Psi \bigcup \{k\}} > EE_{\Psi}^{*}$    \\
\STATE        ~~~~~~ $\Psi={\Psi\bigcup \{k\}}$;           \\    \STATE      ~~~ {\bf{end}} \\
\STATE {\bf{end} }}
\end{algorithmic}
\end{algorithm}
\end{table}}}
 \subsection{Computational Complexity Analysis}

The computational complexity of Algorithm 2 can be evaluated as follows. First, the complexity for obtaining  bandwidth and power allocation variables in Algorithm 1 linearly increases with the number of MUs and the number of SUs, i.e., $\mathcal{O}(K+N)$.  Second, the complexities of the ellipsoid method for updating dual variables \cite{yu2006dual} and the Dinkelbach method for updating $q$ \cite{Boyd,ng2012energy1,loodaricheh2014energy,cheung2012achieving} are both independent of $K$ and $N$. Finally, the complexity of performing the MU selection linearly increases with $K$. Therefore, the total complexity of Algorithm 2 is $\mathcal{O}\big(K(K+N)\big)$ \cite{cormen2009introduction}\footnote{{Note that big $O(\cdot)$ notation is a mathematical notation that is used to illustrate algorithms by how they respond to the changes of the problem size [39]. Thus,  factors that are independent of the problem size $K$ and $N$ are omitted in its formal expression.}}.
 {\subsection{Discussion on Arbitrary Weights}
 As mentioned in Section III, assigning different weights to different SUs in problem (10) does not affect the proposed optimization framework. Now, we show how to tackle problem (10) with arbitrary weights of different SUs.  From (\ref{eq6}),  we know that the weighted system data rate of SUs can be expressed as
\begin{align}\label{eq61}
R_{\rm{total}}& =  \sum_{n=1}^{N}\alpha_n R_n =\sum_{n=1}^{N}\alpha_n\left( r^n_{SC} + \sum_{k=1}^Kx_kr_{k,n}\right)\nonumber\\
&= \sum_{n=1}^N\alpha_n r^n_{SC} + \sum_{n=1}^N\alpha_n\sum_{k=1}^Kx_kr_{k,n} \nonumber\\
&= \sum_{n=1}^N\alpha_n r^n_{SC} + \sum_{k=1}^Kx_k\sum_{n=1}^N\alpha_nr_{k,n},
\end{align}
  where $\alpha_n$ denotes the weight of SU $n$, $\forall\, n$. Due to joint effects of weights and channel conditions, Theorem 1 does not hold any more. More specifically, the bandwidth obtained by the SC via spectrum-power trading from a MU may be shared with multiple SUs rather than one SU in the case of equal weights for SUs. However, the resource allocation algorithm and the MU selection scheme proposed in Section IV and Section V can be readily extended, which are shown as follows. }

{For given MU selection variables $x_k$, $\forall\, k$, the problem transformation between (\ref{eq11c}) and (\ref{eq12c}) can still be applied and it is also easy to verify that the resulting problem in subtractive form is also a concave maximization problem as problem (\ref{eq13c1}). Thus, the joint bandwidth and power allocation can be similarly obtained via exploiting the KKT conditions of the transformed problem. Now, we show how to modify the defined trading EE to tackle the case of arbitrary weights. Recall that the main characteristic of the trading EE in Section V-A is to characterize the obtained throughput and the power consumption in the spectrum-power trading on an MU. Thus, the trading EE of MU $k$,  $\forall\, k\in \mathcal{K}$, in the case of arbitrary weights can be modified as
\begin{align}
EE_{k}=\frac{\sum_{n=1}^N\alpha_nr_{k,n}}{  \sum_{n=1}^{N}\frac{p_{k,n}}{\xi}+\frac{q_{k}}{\xi}}=\frac{\sum_{n=1}^N\alpha_nb_{k,n} \log_2\left(1+\frac{p_{k,n}g_{k,n}}{b_{k,n}N_0}\right)}{  \sum_{n=1}^{N}\frac{p_{k,n}}{\xi}+\frac{q_{k}}{\xi}} ,
\end{align}
where the numerator, $\sum_{n=1}^N\alpha_n r_{k,n}$, is the additional data rate that the SC obtains via serving MU $k$ and the denominator, $ \sum_{n=1}^{N}\frac{p_{k,n}}{\xi}+\frac{q_{k}}{\xi}$, is the total power consumed for both supporting SUs and meeting the QoS of MU $k$. Then, the maximum trading EE can be still readily obtained by solving a counterpart of problem (\ref{eq20}). It is worth noting that the introduction of the weights $\alpha_n$, $\forall\, n$, will not affect the structural properties of the relationship between the system EE of SC and the trading EE, i.e., Lemma 1, (33), and (34) in Appendix D still hold with the modified trading EE. Therefore, Theorem \ref{scheduling} and Corollary \ref{optmality} can be similarly extended. }
\vspace{-0.3cm}
\section{Numerical Results }
\begin{table}[!t]
\centering
\caption{\label{table2} \small{Simulation Parameters}} \label{parameters}
\renewcommand\arraystretch{0.79}
\begin{tabular}{|c|c|}
\hline
{Parameter} & {Description} \\
\hline

Maximum allowed transmit power of the SC, $P^{SC}_{\max}$&      $30$  dBm    {\cite{shakir2013green}}     \\
\hline
Licensed bandwidth of each MU, $W^k_{MC}$  &      $360$  kHz  {\cite{ma2016resource}}\\
\hline
Licensed bandwidth of each SU, $B^{n}_{SC}$ &      $180$  kHz {\cite{ma2016resource}} \\
\hline
Static circuit power of the SC, $P_{\rm{c}}$ &    $2$ W {\cite{shakir2013green}} \\
\hline
Power spectral density of thermal noise  &    $-174$ dBm/Hz \\
\hline
Power amplifier efficiency, $\xi$ &    $0.38$    \\
\hline
Path loss model  &    $(128.1 +37.6\log_{10}d/1000)$  dB    \\
\hline
Lognormal Shadowing  &    $8$ dB    \\
\hline
Penetration loss  &    $20$ dB    \\
\hline
Fading  &    Rayleigh fading    \\
\hline
\end{tabular}
\end{table}

In this section, we provide numerical results to demonstrate the effectiveness of the proposed spectrum-power trading based resource allocation algorithm.
{The main parameters adopted in this work are from relevant works \cite{shakir2013green,ma2016resource,ngo2014joint,saeed2013energy,liu2013massive}.}
 We consider a two-tier heterogeneous network where there exist an MC and an SC with the coverage radii of 500 $m$ and  50 $m$, respectively. Five SUs are uniformly distributed within the coverage of the SC BS while five MUs are uniformly distributed within the distances of [20 200] $m$ away from the SC BS.
  The distance between the SC BS and the MC BS is set to 500 $m$.  Without loss of generality, we assume that all MUs have identical parameters, i.e., the same amount of  available bandwidth, $W^k_{MC}$, and minimum data rate requirement, $R^k_{MC}$. In addition, all  SUs have the identical licensed bandwidth,  $B^n_{SC}$. Unless specified otherwise, the major parameters are listed in Table \ref{parameters} and $R_{SC}$ and $R^k_{MC}$ are set to be 1000 Kbits and 700 Kbits, respectively.
\subsection{System EE versus Maximum Transmit Power of SC, $P^{SC}_{\max}$}
\begin{figure}[!t]
\centering
\includegraphics[width=3.2in]{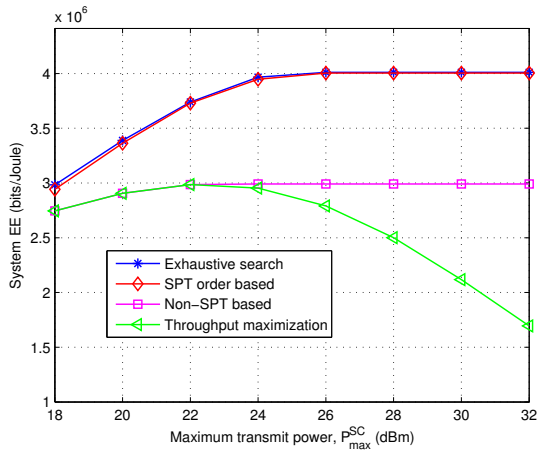}
\caption{System EE versus the maximum allowed transmit power of the SC.}\label{fig2}
\end{figure}

In Figure \ref{fig2}, we compare the achieved system EE  of the following schemes: 1) Exhaustive search \cite{Boyd}; 2) SPT order based: Algorithm 2 in Section V; 3) Non-SPT based: the EE maximization without spectrum-power trading \cite{ng2012energy2}; 4) Throughput Maximization: conventional spectral efficiency maximization \cite{ngo2014joint}.  It is observed that the proposed Algorithm 2 achieves near-optimal performance and outperforms all other suboptimal schemes, which demonstrates the effectiveness of the proposed scheme.
We also observe that the EEs of the SPT order based scheme and the non-SPT based scheme first increases and then remain constants as $P^{SC}_{\max}$ increases. In contrast, the EE of the throughput maximization scheme first increases and then decreases with increasing $P^{SC}_{\max}$, which is due to its greedy use of the transmit power. In addition, it is also seen that the performance gap between the  SPT order based scheme and the non-SPT based scheme first increases and then approaches a constant. This is because when the transmit power of the SC is limited, such as  $P^{SC}_{\max}$ = $12$ dBm, the SC may not have sufficient transmit power freedom to serve many MUs and thereby the spectrum-power trading is less likely realized, which in return limits its own performance improvement. As $P^{SC}_{\max}$ increases, compared with the non-SPT based scheme, the SC not only has more transmit power to improve its EE via serving its own SUs, but also has more transmit power freedom to obtain additional bandwidth from the MC via spectrum-power trading, which thereby strengthens the effect of performance improvement. Finally, when all the `good' MUs with higher trading EE are being scheduled by the SC, then the system EE improves with $P^{SC}_{\max}$ with diminishing return and eventually approaches a constant due to the same reason as that of the non-SPT based scheme.
\begin{figure}[!t]
\centering
\includegraphics[width=3.2in]{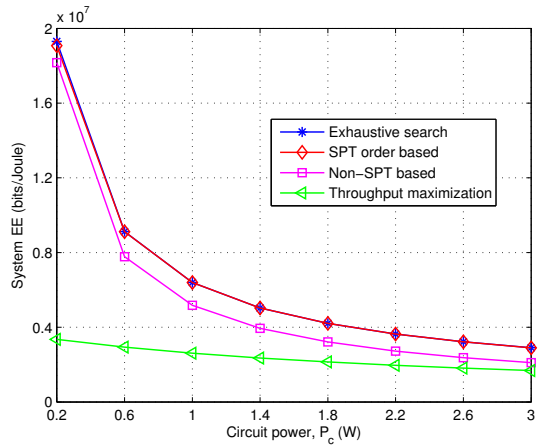}
\caption{System EE versus the circuit power of the SC.}\label{fig3}
\end{figure}
  \subsection{System EE versus Circuit Power of  SC, $P_{\rm{c}}$ }
  Figure \ref{fig3} illustrates the performance of all schemes as a function of the circuit power consumption of the SC. We can observe that the system EE of all schemes decreases with increasing $P_{\rm{c}}$ since the circuit power consumption is always detrimental to the system EE. Also, the proposed Algorithm 2 performs almost the same as the exhaustive search.
  In addition, the performance gap between the non-SPT scheme and the throughput maximization scheme decreases with increasing $P_{{\rm{c}}}$. This is because
 as  $P_{\rm{c}}$ increases, the circuit power consumption dominates the total power consumption rather than the transmit power consumption. Thus, improving the system EE is almost equivalent to improving the system data rate, which only results in marginal performance gap.

 However, it is interesting to note that the performance gap between the SPT order based scheme and the non-SPT based scheme does not decrease but increases when $P_{\rm{c}}$ is in a relatively small regime, such as $P_{\rm{c}} \in [0.2 ~ 1]$ W. This is because when $P_{\rm{c}}$ is very small, the SC system itself enjoys a high system EE which leaves it a less incentive to perform spectrum-power trading with the MC. Thus, the system EE of the SPT order based scheme  decreases with the similar slope as that of the non-SPT based scheme. As $P_{\rm{c}}$ increases, the system EE of the SC further decreases, which would motivate the SC to perform spectrum-power trading. As a result, the performance degradation caused by an increasing $P_{\rm{c}}$ is relieved for the SPT order based scheme, which thereby yields an increased performance gap between these two schemes in small $P_{\rm{c}}$ regime.  Furthermore, when $P_{\rm{c}}$ is sufficiently large such that all the `good' MUs are being selected, the performance gap between these two schemes decreases again due to the the domination of the circuit power in the total power consumption.
\begin{figure}[t]
\centering
\subfigure[Power saved for MC versus distance.]{\includegraphics[width=3.2in, height=2.4in]{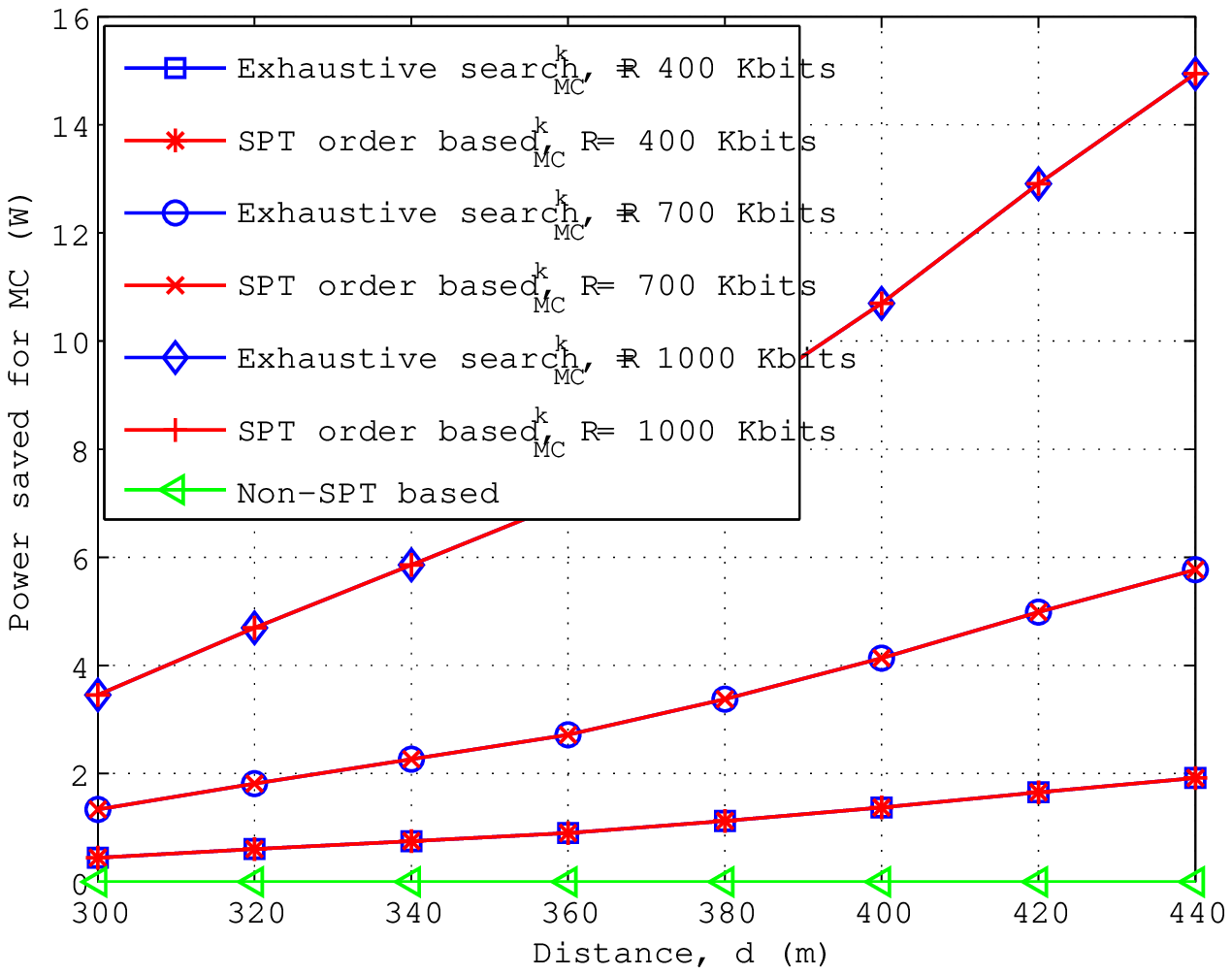}}
\subfigure[System EE of the SC versus distance.]{\includegraphics[width=3.2in,height=2.4in]{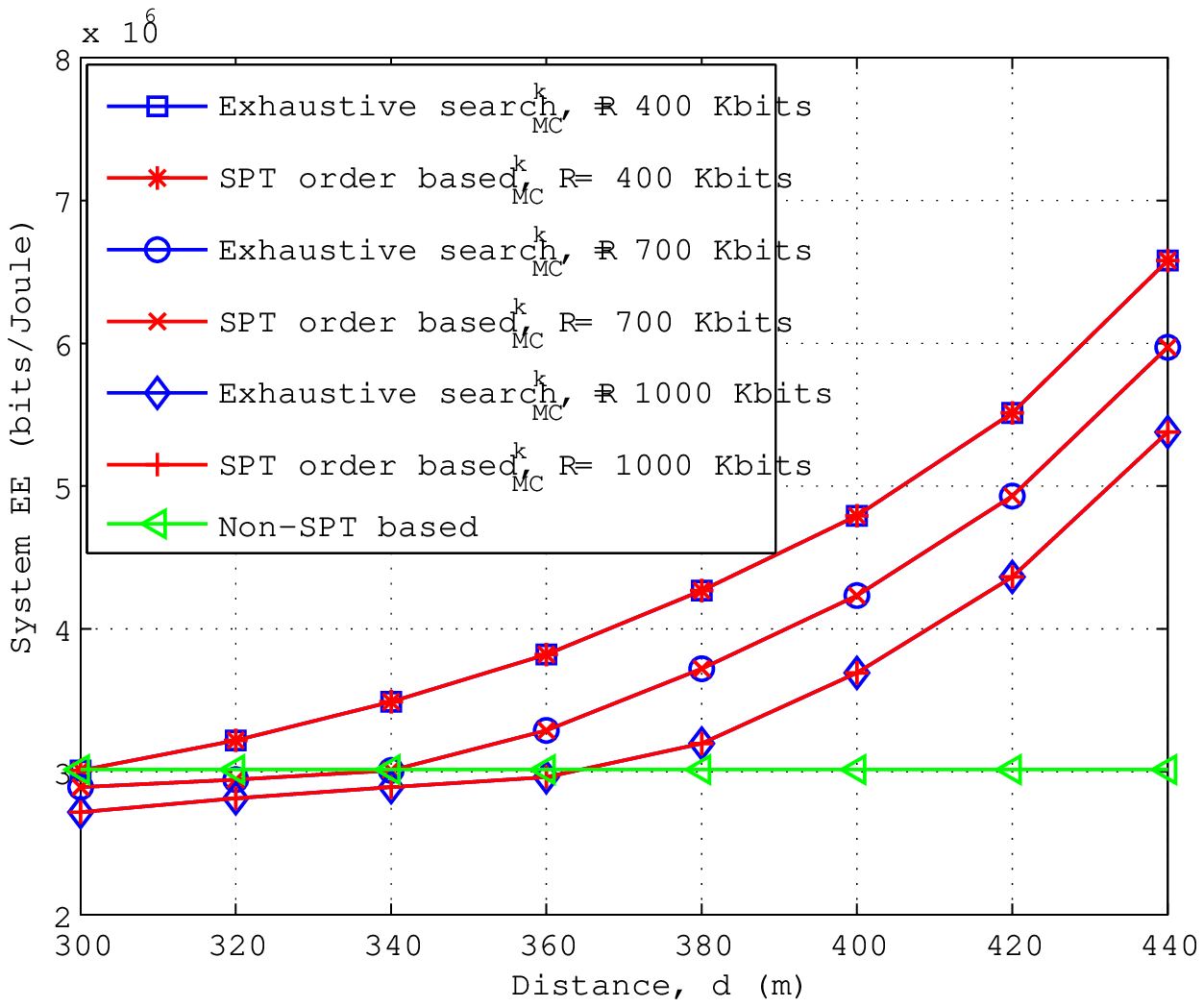}}
\caption{Effect of the distance between MUs and the MC BS on the power saved for MC and the system EE of the SC.} \label{fig6}
\end{figure}
\subsection{Effect of Distance between MUs and MC BS}
In Figure \ref{fig6}, we evaluate the performances of the exhaustive search scheme, SPT order based scheme, and the non-SPT based scheme versus the distance between MUs and the MC BS. Without loss of generality, we assume that all MUs are located at the same distance from the SC BS and the SC BS helps to serve all MUs from the MC, i.e., $x_k=1$, $\forall\, k$, where the MU selection is not performed.  Thus, the SPT order based scheme performs the same as the exhaustive search.
In Figure  \ref{fig6} (a), we can see that under the fixed minimum data rate requirements, more transmit power consumption is saved via the proposed spectrum-power trading when MUs are farther away from the MC BS. In addition, when the the MU data rate requirements are higher, it also saves more transmit power consumption for the MC BS. These implies that the proposed spectrum-power trading is effective by offloading the MUs to the SC BS, especially when the MUs are located in the cell edge area while requiring high user data rates. In contrast, in Figure  \ref{fig6} (b), we illustrate the system EE of the SC BS versus the distance between MUs and the MC BS. Basically, when the distance is larger, the proposed scheme enables higher system EE gain for the SC BS. However, as the MU data rate requirements, $R^k_{MC}$,  increase, the system EE of the SC BS decreases since it either obtains less bandwidth or costs more transmit power by serving these MUs. In particular, when the MUs are farther away from the SC BS and also require higher user data rates, the achieved system EE may even be lower than the system EE without MU offloading. This means that although the spectrum-power trading benefits the MC,
the MU selection is necessary to improve the system EE from the perspective of the SC.
\subsection{Effect of Licensed Bandwidth of MUs, $W^k_{MC}$}

\begin{figure}[t]
\centering
\subfigure[Number of selected MUs versus $W^k_{MC}$.]{\includegraphics[width=3.2in, height=2.4in]{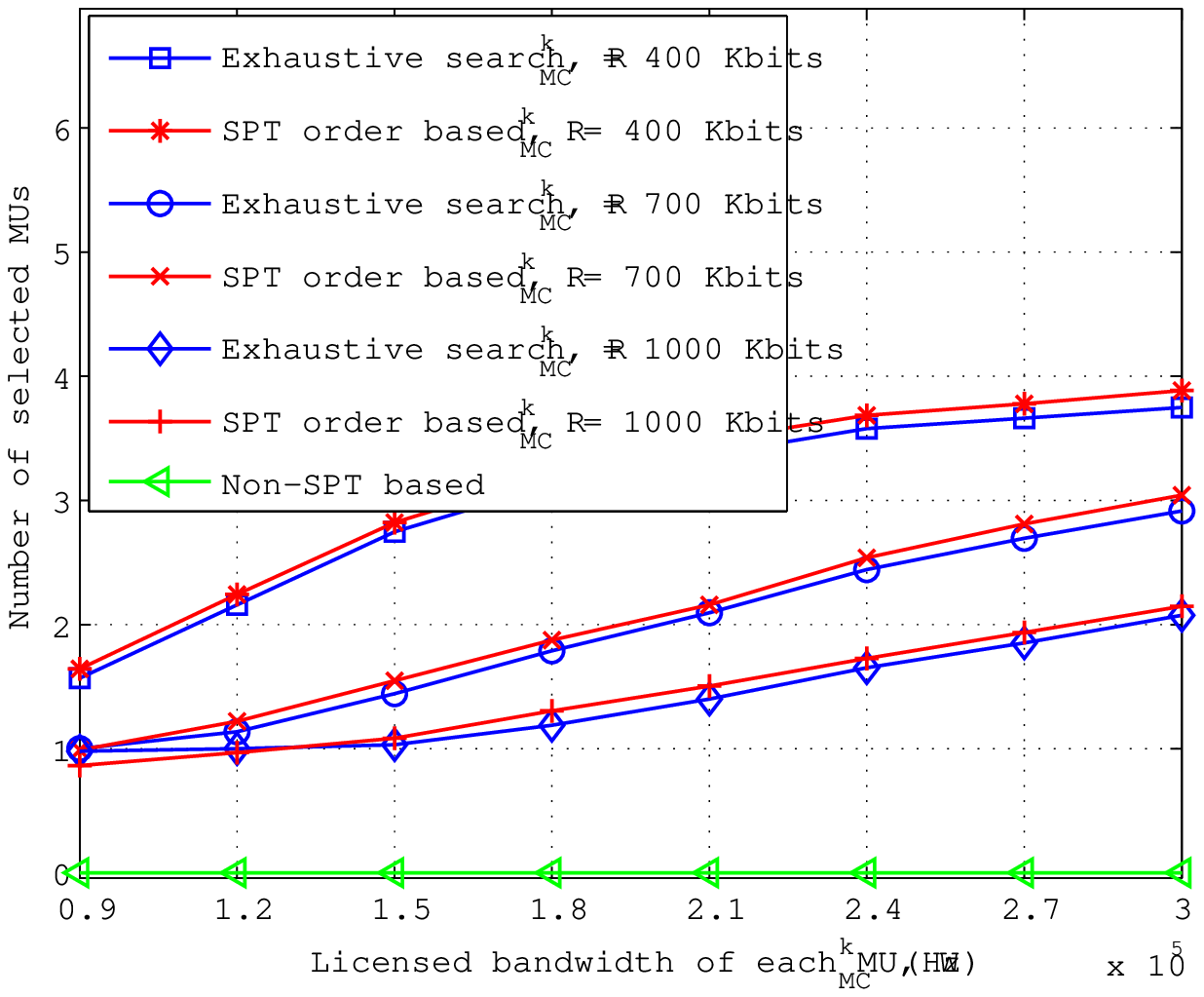}}
\subfigure[System EE of the SC versus $W^k_{MC}$.]{\includegraphics[width=3.2in,height=2.4in]{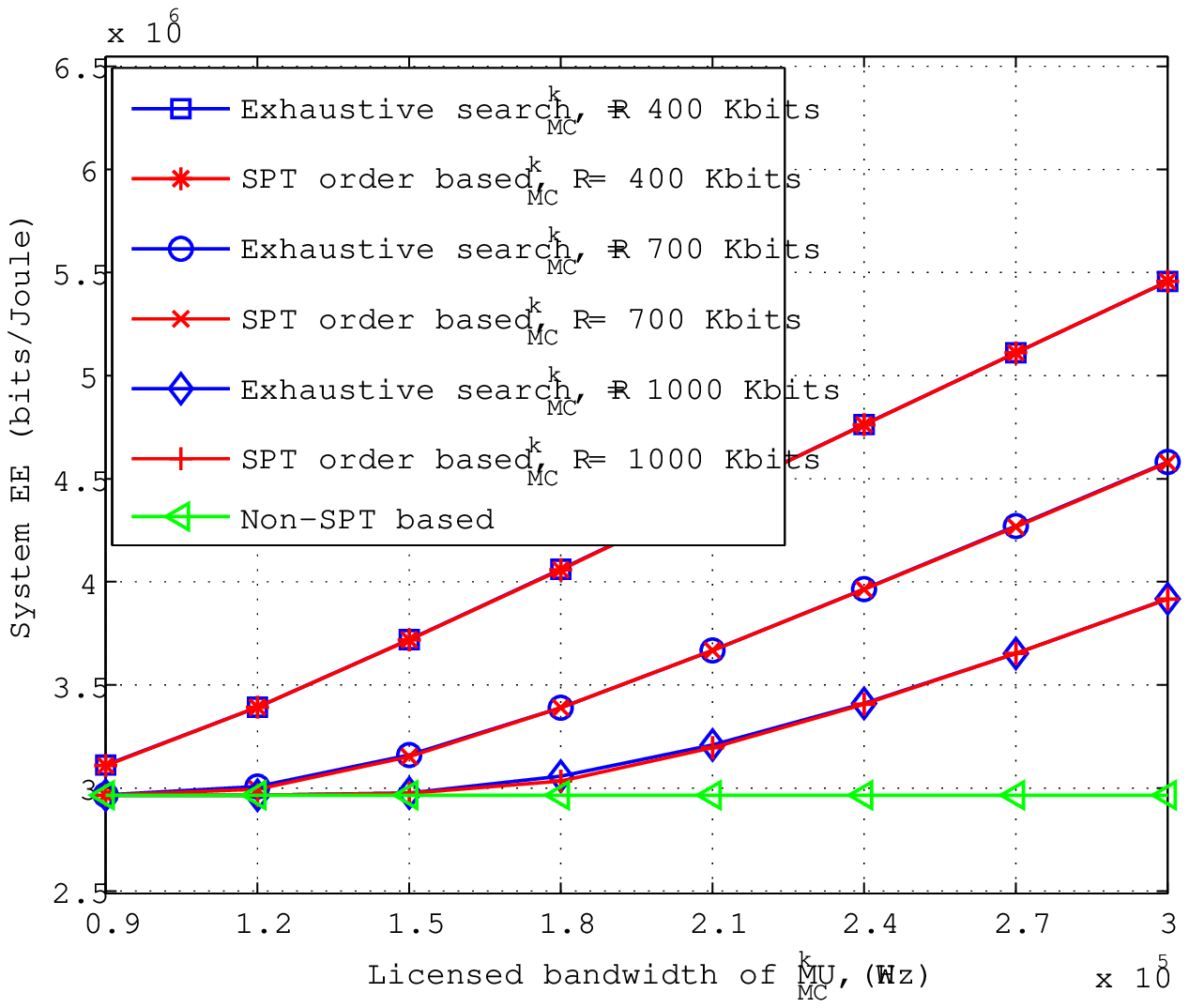}}
\caption{Effect of licensed bandwidth of MUs on the number of selected MUs and the system EE of the SC.} \label{fig7}
\end{figure}

In Figure \ref{fig7}, we we evaluate the effect of licensed bandwidth on the number of MUs selected by the SC as well as on the system EE of the SC. Specifically, in Figure \ref{fig7} (a), it is observed that the proposed SPT order based scheme still achieves an excellent  performance which further demonstrates the effectiveness of exploiting the trading EE for MU selection.
In addition, we can also find that the number of MUs selected by the SC increases with an increasing $W^k_{MC}$ under a fixed user data rate requirement. This also coincides with our theoretical analysis for trading EE in Section V:  due to the monotonically increasing characteristic of the trading EE with respect to $W^k_{MC}$,
an MU with more bandwidth provided by the MC achieves a higher trading EE such that this MU is more likely to satisfy the MU selection condition. In contrast, for a given $W^k_{MC}$, requiring higher user data rates provides less incentives for the SC to serve MUs and thus the number of MUs selected by the SC decreases with a more stringent $R^k_{MC}$.
In Figure \ref{fig7} (b), it is clear to see that the system EE increases with an increasing $W^k_{MC}$. The performance improvement comes from two aspects. First, given a fixed minimum data rate requirement of the MU, the more the bandwidth provided by the MC, the less the power consumed by the SC via spectrum-power trading, which thus helps to improve the system EE of the SC.  Second, as mentioned, a larger bandwidth will motivate the SC to serve more MUs and in return, to allow the SC to obtain more additional bandwidth  via spectrum-power trading.

\section{Conclusions}
In this paper, we investigated the spectrum-power trading between an SC and an MC to improve the system EE of the SC as well as reducing the power consumption of the MC. Specifically, MU selection, bandwidth allocation, and power allocation were jointly optimized while guaranteeing the QoS of both networks. The system EE maximization problem was first simplified by showing that the bandwidth from each MU is only shared with at most one SU in the SC. Given the MU selection, we transformed the fractional-form optimization problem into a substractive one that can be solved efficiently with optimality. Then, we proposed a trading EE based MU selection scheme by studying the intrinsic relationship between the trading EE of an MU and the system EE.
 {Simulation results showed that the proposed algorithm obtains close-to-optimal performance and  also demonstrated the performance gains achieved by the proposed spectrum-power trading scheme for both the SC and the MC, especially when MUs are far away from the MC BS. {For future work, although different weights can be assigned to different SUs to achieve a notion of fairness, it is still worth investigating the EE maximization problem with individual QoS constraints for SUs explicitly.}

\appendices
\section{Proof of Theorem 1}
{We first prove that the bandwidth of MU $m$ is shared with at most one SU who has the largest channel power gain on the bandwidth of MU $m$. Here, we use ``at most one" instead of ``only one" is because this MU may also be rejected by the SC in terms of spectrum-power trading, and thus the bandwidth of this MU may not be shared with any SU. The proof is shown by contradiction as follows.} Assume that  $S =\Big\{\{x^*_k\},\{p^*_n\},\{p^*_{k,n}\}, \{b^*_{k,n}\}, \{q^*_k\}, \{w^*_k\}\Big\}$ achieves the optimal solution of problem (\ref{eq16}) and there exist an MU $m$ whose bandwidth has been shared with two SUs, SU $m'$ and SU $\ell$, $\ell\neq m'$, in the SC, i.e.,  $b^*_{m, m'}>0$, $p^*_{m,m'}>0$ and $b^*_{m, \ell}>0$, $p^*_{m, \ell}>0$. {Denote SU $m'$ is the SU that has the largest channel power gain on the bandwidth of MU $m$, i.e., $g_{m,m'}>g_{m,n}, \forall\, n\neq m.$ Note that the probability of two SUs that have the same channel power gain is zero due to the continuity and the randomness of the channel fading.}
 Then, we construct a different solution $\widehat{S} = \Big\{\{\widehat{x}_k\},\{\widehat{p}_n\},\{\widehat{p}_{k,n}\}, \{\widehat{b}_{k,n}\}, \{\widehat{q}_k\}, \{\widehat{w}_k\}\Big\}$ where $\{\widehat{x}_k\}=\{x^*_k\}$, $\{\widehat{p}_n\}= \{p^*_n\}$,    $\{\widehat{q}_k\}=\{q^*_k\}$, $\{\widehat{w}_k\}= \{w^*_k\}$, and
\begin{align}\label{eq5}
\widehat{b}_{k,n} =\left\{
\begin{array}{lclcl}
b^*_{m, m'} + b^*_{m,\ell},& k=m, n = m',\\
0,& k=m, n\neq m',\\
b^*_{k,n},&  k\neq m, n\in \mathcal{N}.
\end{array}\right.
\end{align}
\begin{align}\label{eq5}
\widehat{p}_{k,n} =\left\{
\begin{array}{lclcl}
p^*_{m, m'} + p^*_{m,\ell},& k=m, n = m',\\
0,& k=m, n\neq m',\\
p^*_{k,n},&   k\neq m, n\in \mathcal{N}.
\end{array}\right.
\end{align}
We also note that the constructed solution satisfies all the constraints in problem (\ref{eq16}) and is thereby a feasible solution. Since the case when $x^*_k=0$ is obviously satisfied, we only discuss the case when $x^*_k=1$ in the following.
Then, the data rate of MU $m$ brought for the SC can be expressed as
\begin{align}
\widehat{b}_{m,m'}&\log_2\left(1+\frac{\widehat{p}_{m,m'}g_{m,m'}}{\widehat{b}_{m,m'}N_0}\right)
 =(b^*_{m,m'}+b^*_{m,\ell})\log_2\left(1+\frac{(p^*_{m,m'}+p^*_{m,\ell})g_{m,m'}}{(b^*_{m,m'}+b^*_{m,\ell})N_0}\right) \nonumber \\
&\overset{(a)}{\geq}{b}^*_{m,n}\log_2\left(1+\frac{{p}^*_{m,m'}g_{m,m'}}{{b}^*_{m,m'}N_0}\right) +{b}^*_{m,\ell}\log_2\left(1+\frac{{p}^*_{m,\ell}g_{m,m'}}{{b}^*_{m,\ell}N_0}\right) \nonumber \\
&\overset{(b)}> {b}^*_{m,n}\log_2\left(1+\frac{{p}^*_{m,m'}g_{m,m'}}{{b}^*_{m,m'}N_0}\right) +{b}^*_{m,\ell}\log_2\left(1+\frac{{p}^*_{m,\ell}g_{m,\ell}}{{b}^*_{m,\ell}N_0}\right),
\end{align}
where  inequality $(a)$ holds due to the concavity of $f \log_2(1+\frac{y}{f})$ and strict inequality $(b)$ holds due to $g_{m,m'}>g_{m,\ell}$, $\ell\neq m'$.
This means that the constructed solution $\widehat{S}$ achieves higher system data rate with the same total power consumption and thus yields higher system EE than $S^*$ which contradicts the assumption that $S^*$ is optimal.

Now, we show that constraints C2 and C3 are met with equalities and the proof is summarized as follows. 1) if $x_k=0$, then from C2, any feasible solution of problem (\ref{eq16}) must satisfy $b_{k,k'}+w_k\leq 0$. Since $b_{k,k'}\geq0$ and $w_k\geq0$, it follows that $b_{k,k'}=0$ and $w_k=0$. Thus, C2 and C3 are met with equalities. 2) if $x_k=1$ and  $b_{k,n} +w_{k}< W^k_{MC}$ and (or)  $w_k\log_2\left(1+\frac{q_kh_k}{w_kN_0}\right)< R^k_{MC}$ holds in the optimal solution, we can always construct another solution by increasing $b_{k,k'}$ and (or) decreasing $p_k$ such that C2 and C3 are met with equalities while achieving a larger system EE, which contradicts that the optimal solution is achieved under strict inequality constraints C2 and C3.
 Theorem 1 is thus proved.
\section{Proof of Theorem \ref{theorem70}}
Taking the partial derivative of $\mathcal{L}$ with respect to $p_n$, $p_{k,k'}$, and $w_k$, respectively, yields
\begin{align}
\frac{\partial\mathcal{L}}{\partial {p}_{n}}&=\frac{(1+\mu)B^n_{SC}g_{n}}{(B^n_{SC}N_0+{p}_{n}g_{n})\ln2}-\left({\frac{q}{\xi}+\lambda}\right), \forall\, n\in \mathcal{N}, \label{apdx_eq1}\\
\frac{\partial\mathcal{L}}{\partial {p}_{k,k'}}&=\frac{(1+\mu)(W^k_{MC}-w_k)g_{k,k'}}{((W^k_{MC}-w_k)N_0+{p}_{k,k'}g_{k,k'})\ln2}-\left({\frac{q}{\xi}+\lambda}\right), \forall\, k\in \Psi, \label{apdx_eq1}\\
\frac{\partial\mathcal{L}}{\partial w_{k}}&=-(1+\mu)\log_2\left(1+\frac{{p}_{k,k'}g_{k,k'}}{(W^k_{MC}-w_k)N_0}\right)+
                                                                                              \frac{(1+\mu){p}_{k,k'}g_{k,k'}}{\left((W^k_{MC}-w_k)N_0+{p}_{k,k'}g_{k,k'}\right)\ln2}
   \nonumber\\
                                                                                              &-\left(\frac{q}{\xi}+\lambda\right)\left(\left(2^\frac{R^{k}_{MC}}{w_k}-1\right)\frac{N_0}{h_k}- 2^\frac{R^{k}_{MC}}{w_k}\frac{R^{k}_{MC}N_0}{w_kh_k}\ln2 \right), \forall\, k\in \Psi.
\end{align}
Setting $\frac{\partial\mathcal{L}}{\partial {p}_{k,k'}}=0$ and $\frac{\partial\mathcal{L}}{\partial p_n}=0$, the optimal transmit power ${p}_{k,k'}$ and $p_n$ can be obtained as (\ref{eq4.25}) and (\ref{eq4.26}), respectively.
Substituting (\ref{eq4.25})  into $\frac{\partial\mathcal{L}}{\partial w_{k}}$ yields
\begin{align}
\frac{\partial\mathcal{L}}{\partial w_k}& =-(1+\mu)\log_2\left(1+{{\widetilde{p}}_{k,k'}\frac{g_{k,k'}}{N_0}}\right)+
                                                                                           \left(\frac{q}{\xi}+\lambda\right)\widetilde{p}_k
   \nonumber\\
                                                                                              &-\left(\frac{q}{\xi}+\lambda\right)\left(\left(2^\frac{R^{k}_{MC}}{w_k}-1\right)\frac{N_0}{h_k}- 2^\frac{R^{k}_{MC}}{w_k}\frac{R^{k}_{MC}N_0}{w_kh_k}\ln2 \right), \forall\, k\in \Psi,      \label{apdx_eq5}
\end{align}
where ${\widetilde{p}}_{k,k'} = \left[\frac{(1+\mu)\xi}{({q}+\lambda\xi)\ln2}-\frac{N_0}{g_{k,k'}}\right]^+$.
Note that $\frac{\partial\mathcal{L}}{\partial w_k}$ now only involves the optimization variable $w_k$. Setting $\frac{\partial\mathcal{L}}{\partial w_k}=0$, we have 
\begin{align}\label{eq31}
\left(2^\frac{R^{k}_{MC}}{w_k}\frac{R^{k}_{MC}N_0}{w_kh_k}\ln2 - \left(2^\frac{R^{k}_{MC}}{w_k}-1\right)\frac{N_0}{h_k} \right) = \frac{\mathcal{C}}{\frac{q}{\xi}+\lambda}, \forall\, k \in \Psi,
\end{align}
where $\mathcal{C}= (1+\mu)\log_2\left(1+{{\widetilde{p}}_{k,k'}\frac{g_{k,k'}}{N_0}}\right)-\left(\frac{q}{\xi}+\lambda\right)\widetilde{p}_k$ and (\ref{eq4.22})  is obtained from (\ref{eq31}). In addition, it is easy to verify that the left hand side of (\ref{eq31}) is a monotonically decreasing function of $w_k$, which implies that there exists a unique $w_k$ that satisfies (\ref{eq31}).  Thus, the value of $w_k$ can be efficiently obtained by the bisection method.
\section{Proof of Theorem \ref{trading}}\label{apdix1}
Given $w_k$ in problem (\ref{eq20}),  it is easy to see  that $EE_k$ increases with $b_{k,k'}$ and decreases with $q_k$. Thus, it can be verified that C2 and C4 are met with equalities at the optimal solution.
Substituting $b_{k,k'} =W^k_{MC}-w_k$ and $q_k= \left(2^\frac{R^k_{MC}}{w_k}-1\right)\frac{w_kN_0}{h_k}$ into problem (\ref{eq20}) results in  (\ref{eq201}). In addition, since $(W^k_{MC}- w_k)\log_2\left(1+\frac{p_{k,k'}g_{k,k'}}{(W^k_{MC}- w_k)N_0}\right)$ is strictly concave over $w_k$ and $\frac{p_{k,k'}}{\xi}+\left(2^\frac{R^k_{MC}}{w_k}-1\right)\frac{w_kN_0}{h_k\xi}$ are jointly convex over $p_{k,k'}$ and $w_k$, then it follows that the objective function of problem (\ref{eq201}), $EE_k$,  is jointly quasi-concave with respect to  $p_{k,k'}$ and $w_k$ \cite{Boyd, miao2010energy}, which completes the proof of Theorem 3.
\section{Proof of Theorem \ref{scheduling}}\label{apdix1}
{We first introduce a lemma  \cite{qing15_wpcn_twc} to facilitate the proof.
\begin{lemma}
Assume that  $a$, $b$, $c$, and $d$ are arbitrary positive numbers. Then, we have
\begin{align}\label{adx_eq1}
\min{\left\{\frac{a}{b},\frac{c}{d}\right\}} \leq\frac{a+c}{b+d}\leq\max{\left\{\frac{a}{b},\frac{c}{d}\right\}},
\end{align}
 where ``=" holds if and only if $\frac{a}{b}=\frac{c}{d}$.
\end{lemma}

Based on Lemma 1, we first prove 1) in Theorem \ref{scheduling}.
{Let $\mathcal{S^*}=\Big\{\{{p}^*_{n}\}, \{{p}^*_{k,k'}\}, \{{q}^*_k\}, \{{b}^*_{k,k'}\}, \{{w}^*_k\}\Big\}$ denote the optimal solution of problem (\ref{eq20}) and its corresponding user EE is denoted as $EE_k$. Let $\mathcal{\widehat{S}}=\Big\{\{{p}^*_{n}\}, \{\widehat{p}_{k,k'}\}, \{\widehat{q}_k\},$  $\{\widehat{b}_{k,k'}\}, \{\widehat{w}_k\}\Big\}$ and $\mathcal{\widetilde{S}}=\Big\{\{\widetilde{{p}}_{n}\}, \{\widetilde{p}_{k,k'}\}, \{\widetilde{q}_k\}, \{\widetilde{b}_{k,k'}\},\{\widetilde{w}_k\}\Big\}$ denote the optimal solutions of problem (\ref{eq16}) with $x_k=1$ for $k\in \Psi$ and  $k\in \Psi\bigcup \{m\}$, respectively,  where $m\notin \Psi$.} The corresponding system EEs are denoted as $EE^*_{\Psi}$ and $EE^*_{\Psi\bigcup \{m\}}$, respectively.
Then, we have the following
\begin{eqnarray}\label{eq39}
EE^*_{\Psi\bigcup \{m\}}
&=& \frac{\sum_{n=1}^Nr^n_{SC}(\widetilde{p}_{n})+\sum_{k\neq m}r_{k,k'}(\widetilde{b}_{k,k'},\widetilde{p}_{k,k'}) + r_{m,m'}(\widetilde{b}_{m,m'},\widetilde{p}_{m,m'})}
{\sum_{n=1}^N\frac{\widetilde{p}_n}{\xi} + \sum_{k\neq m}\frac{\widetilde{p}_{k,k'}}{\xi}+\sum_{k\neq m}\frac{\widetilde{q}_k}{\xi}+P_{\rm{c}} + \frac{\widetilde{p}_{m,m'}}{\xi}+ \frac{\widetilde{q}_m}{\xi}} \nonumber  \\
&\overset{(a)} \geq& \frac{\sum_{n=1}^Nr^n_{SC}(\widehat{p}_{n})+\sum_{k\neq m}r_{k,k'}(\widehat{b}_{k,k'},\widehat{p}_{k,k'}) + r_{m,m'}({b}^*_{m,m'},{p}^*_{m,m'})}
{\sum_{n=1}^N\frac{\widehat{p}_n}{\xi} + \sum_{k\neq m}\frac{\widehat{p}_{k,k'}}{\xi}+\sum_{k\neq m}\frac{\widehat{q}_k}{\xi}+P_{\rm{c}} + \frac{{p}^*_{m,m'}}{\xi}+ \frac{{q}^*_m}{\xi}}  \nonumber  \\
&\overset{(b)} \geq& \min\left\{\frac{\sum_{n=1}^Nr^n_{SC}(\widehat{p}_{n})+\sum_{k\neq m}r_{k,k'}(\widehat{b}_{k,k'},\widehat{p}_{k,k'}) }
{\sum_{n=1}^N\frac{\widehat{p}_n}{\xi} + \sum_{k\neq m}\frac{\widehat{p}_{k,k'}}{\xi}+\sum_{k\neq m}\frac{\widehat{q}_k}{\xi}+P_{\rm{c}}},
\frac{ r_{m,m'}({b}^*_{m,m'},{p}^*_{m,m'})}{\frac{p^*_{m,m'}}{\xi}+ \frac{q^*_m}{\xi}} \right\}    \nonumber \\
&=& \min\left\{EE^*_{\Psi}, EE^*_{m}\right\},
\end{eqnarray}
where inequality $(a)$ holds due to the fact that $\widetilde{S}$ is the optimal solution of problem (\ref{eq16}) with $x_k=1$ for $k\in \Psi\bigcup \{m\}$. Inequality  $(b)$ holds due to  Lemma 1 and the equality ``='' holds only when $EE^*_{\Psi}= EE^*_{m}$. Thus, we can conclude $EE^*_{m} > EE^*_{\Psi}$ $\Longrightarrow$ $EE^*_{\Psi\bigcup \{m\}} > EE^*_{\Psi}$, which completes the proof of the ``if'' part. In the next, we prove $EE^*_{\Psi\bigcup \{m\}}>EE^*_{\Psi}$ $\Longrightarrow$ $EE^*_m> EE^*_{\Psi}$, which is equivalent to its contrapositive proposition, i.e., $EE^*_{m} \leq EE^*_{\Psi}$ $\Longrightarrow$ $EE^*_{\Psi\bigcup \{m\}} \leq EE^*_{\Psi}$.
 Then, we have the following
 \begin{eqnarray}\label{eq391}
EE^*_{\Psi\bigcup \{m\}}
&=& \frac{\sum_{n=1}^Nr^n_{SC}(\widetilde{p}_{n})+\sum_{k\neq m}r_{k,k'}(\widetilde{b}_{k,k'},\widetilde{p}_{k,k'}) + r_{m,m'}(\widetilde{b}_{m,m'},\widetilde{p}_{m,m'})}
{\sum_{n=1}^N\frac{\widetilde{p}_n}{\xi} + \sum_{k\neq m}\frac{\widetilde{p}_{k,k'}}{\xi}+\sum_{k\neq m}\frac{\widetilde{q}_k}{\xi}+P_{\rm{c}} + \frac{\widetilde{p}_{m,m'}}{\xi}+ \frac{\widetilde{q}_m}{\xi}} \nonumber  \\
&\overset{(c)} \leq& \max\left\{\frac{\sum_{n=1}^Nr^n_{SC}(\widetilde{p}_{n})+\sum_{k\neq m}r_{k,k'}(\widetilde{b}_{k,k'},\widetilde{p}_{k,k'}) }
{\sum_{n=1}^N\frac{\widetilde{p}_n}{\xi} + \sum_{k\neq m}\frac{\widetilde{p}_{k,k'}}{\xi}+\sum_{k\neq m}\frac{\widetilde{q}_k}{\xi}+P_{\rm{c}}},
\frac{r_{m,m'}(\widetilde{b}_{m,m'},\widetilde{p}_{m,m'})}
{\frac{\widetilde{p}_{m,m'}}{\xi}+ \frac{\widetilde{q}_m}{\xi}}\right\}    \nonumber \\
&\overset{(d)} \leq& \max\left\{\frac{\sum_{n=1}^Nr^n_{SC}(\widehat{p}_{n})+\sum_{k\neq m}r_{k,k'}(\widehat{b}_{k,k'},\widehat{p}_{k,k'}) }
{\sum_{n=1}^N\frac{\widehat{p}_n}{\xi} + \sum_{k\neq m}\frac{\widehat{p}_{k,k'}}{\xi}+\sum_{k\neq m}\frac{\widehat{q}_k}{\xi}+P_{\rm{c}}},
\frac{ r_{m,m'}({b}^*_{m,m'},{p}^*_{m,m'})}{\frac{p^*_{m,m'}}{\xi}+ \frac{q^*_m}{\xi}} \right\}    \nonumber \\
&=& \max\left\{EE^*_{\Psi}, EE^*_{m}\right\},
\end{eqnarray}
where inequality $(c)$ holds due to Lemma 1 and the equality ``='' represents the special case when the current SC EE is the same as trading EE of MU $k$.  Inequality $(d)$ holds due to the fact that both ${\widehat{S}}$ and ${S}^*$  are optimal solutions of problem (\ref{eq16}) with $x_k=1$ for $k\in \Psi$ and problem (\ref{eq20}), respectively. Thus, if $EE^*_{m} \leq EE^*_{\Psi}$, then we can conclude $EE^*_{\Psi\bigcup \{m\}} \leq EE^*_{\Psi}$ from (\ref{eq391}), which completes the proof of the ``only if" part.

Based on 1), we next prove 2) in Theorem \ref{scheduling}. When the minimum system data rate constraint C4 instead of C1 is considered in problem (\ref{eq16}), the inequality $(d)$ may not hold in (\ref{eq391}). This is because $\widetilde{S}$ only needs to satisfy
\begin{align}
&\sum_{n=1}^Nr^n_{SC}(\widetilde{p}_n) + \sum_{k=1}^Kr_{k,k'}(\widetilde{b}_{k,k'}, \widetilde{p}_{k,k'}) \nonumber \\
=&\sum_{n=1}^Nr^n_{SC}(\widetilde{p}_n) + \sum_{k\neq m}r_{k,k'}(\widetilde{b}_{k,k'}, \widetilde{p}_{k,k'}) + r_{m,m'}(\widetilde{b}_{m,m'}, \widetilde{p}_{m,m'})   \geq R^{SC}_{\mathop{\min}},  \label{eq_apdx38}
\end{align}
while $\widehat{S}$ has to satisfy
\begin{align}
\sum_{n=1}^Nr^n_{SC}({\widehat{p}}_n) + \sum_{k\neq m}r_{k,k'}(\widehat{b}_{k,k'}, \widehat{p}_{k,k'})  \geq R^{SC}_{\mathop{\min}}. \label{eq_apdx39}
\end{align}
From (\ref{eq_apdx38}) and (\ref{eq_apdx39}), we note that the feasible transmit power region of  $\widetilde{S}$ is larger than the feasible region composed of $S^*$ and $\widehat{S}$, which leads that the inequality $(d)$ may not hold in (\ref{eq391}).  However, based on this,  it is straightforward to show that inequality $(a)$ still holds in (\ref{eq39}).

Based on 1), we next prove 3) in Theorem \ref{scheduling}. When the total power constraint C1 instead of C4 is considered in problem (\ref{eq16}), the inequality $(a)$ may not hold in (\ref{eq39}). This is because the solutions $S^*$ and $\widehat{S}$ are restricted to individual total power constraints, i.e.,
 \begin{align}
\sum_{n=1}^{N}{\widehat{p}_{n}} + \sum_{k\neq m}\widehat{p}_{k,k'}+\sum_{k\neq m}\widehat{q}_k\leq P^{SC}_{\mathop{\max}},  \label{eq_apdx35}\\
{p}^*_{m,m'}\geq 0, {q}^*_m\geq 0,  \label{eq_apdx36}
 \end{align}
 while $\widetilde{S}$ only have one total power constraint,
  \begin{align}
\sum_{n=1}^{N}{\widetilde{p}_{n}} +\sum_{k=1}^K\widetilde{p}_{k,k'} +\sum_{k=1}^{K}\widetilde{q}_k =\sum_{n=1}^{N}{\widetilde{p}_{n}}  + \sum_{k\neq m}\widetilde{p}_{k,k'}+\sum_{k\neq m}\widetilde{q}_k + \widetilde{p}_{m,m'}+\widetilde{q}_m \leq P^{SC}_{\mathop{\max}}. \label{eq_apdx37}
 \end{align}
From (\ref{eq_apdx35}), (\ref{eq_apdx36}), and (\ref{eq_apdx37}), we note that the feasible transmit power region composed of $S^*$ and $\widehat{S}$ is larger than the feasible region of  $\widetilde{S}$, which leads that the inequality $(a)$ may not hold in (\ref{eq39}). However, based on this, it is straightforward to show that inequality $(d)$ still holds in (\ref{eq391}).

\section{Proof of Corollary \ref{optmality}}\label{apdix1}
Since the MUs are sorted in the descending order in terms of the trading EE, i.e., $EE_{1}^{*} > EE_{2}^{*}>,...,> EE_{K}^{*}$, we have the following lemma in the absence of constraints C1 and C4.
\begin{lemma}
1) If selecting MU $k$ increases the system EE of the SC, i.e., $EE_{\Psi}< EE_k$,  then selecting MU $\ell $, $\forall\, \ell\leq k$, also increases the system EE of the SC. 2) If selecting MU $k$ decreases the system EE of the SC, i.e., $EE_{\Psi}>EE_k$, then selecting MU $\ell $, $\forall\, \ell\geq k$, also decreases the system EE of the SC.
\end{lemma}
\begin{proof}
 If $EE_k> EE_{\Psi}$, then we have $EE_{\ell}\geq EE_k >EE_{\Psi}$, $\forall\, \ell\leq k$,  due to the descending order of MUs. Since $EE_{\ell}> EE_{\Psi}$ has been proved as the sufficient and necessary for selecting MU $\ell$ in Theorem \ref{scheduling}, we have the first statement. If  $EE_k< EE_{\Psi}$, then we have $EE_{\ell}\leq EE_k< EE_{\Psi}$, $\forall\, \ell \geq k$, which results the second statement.
\end{proof}

 From Lemma 2, it is easy to prove that there exists an MU $k^*$, for $0 \leq k^*\leq K$  such that the system EE of SC
 increases with $k$ for $0\leq k\leq k^*$ and decreases with $k$ for $k^* \leq k \leq K$, respectively.  As special cases, $k^*= 0$ or $k^*=K$ means the system EE without spectrum-power trading or with spectrum-power trading for all MUs. Thus, we have the following corollary which can be easily proved based on previous the above discussion.  
\begin{corollary}\label{set_alg2}
 MUs only with order index $0\leq k\leq k^*$, $\forall\, 0  \leq k^*\leq K$, are selected by Algorithm 2.
\end{corollary}

For the purpose of illustration, we denote $\Psi=\{0\}$ as the case when no MU is selected by the SC for spectrum-power trading.
With Corollary \ref{set_alg2}, we only need to prove that the system EE of the SC achieved based on MU set $\Psi^*= \{0, 1, ..., k^*\}$ is larger than that of any other set $\Psi$, which is shown by contradiction. Without loss of generality, we assume that $\widehat{\Psi}$ is the optimal MU set but there exists an MU $m$ for $m\leq k^*$ that does not belong to $\widehat{\Psi}$ and  an MU $n$ for $n> k^*$ that belongs to $\widehat{\Psi}$, i.e., $\widehat{\Psi} = \{0, ..., m-1,m+1,..., k^*, n\}$. All other cases can be directly extended from the study of this assumption. Thus, we only need to show $EE^*_{\Psi^*}>EE^*_{\widehat{\Psi}}$.  We introduce an auxiliary MU set $\widetilde{\Psi}= \{0, 1,..., k^*, n\}$. Since MU $k^*$ and $(k^*+1)$ is selected and not selected by Algorithm 2, respectively,
 from Theorem \ref{scheduling}, we have  $EE^*_{\Psi^*}<EE^*_{k^*}\leq EE^*_{m} $ and $EE^*_{\Psi^*}>EE^*_{k^*+1}\geq EE^*_{n}$, respectively, due to the descending order.  With $EE^*_{\Psi^*}> EE^*_{n}$, it follows that   $EE^*_{\Psi^* \bigcup\{n\} }=EE^*_{\widetilde{\Psi}} < EE^*_{\Psi^*}<EE^*_{k^*}\leq EE^*_{m}$. Then, with $EE^*_{\widetilde{\Psi}}< EE^*_{m}$, it follows that  $EE^*_{\widetilde{\Psi}\backslash m}= EE^*_{\widehat{\Psi}}<EE^*_{\widetilde{\Psi}}$, which contradicts the assumption that $\widehat{\Psi}$ is the optimal MU set. Corollary 1 is thus proved.
\bibliographystyle{IEEEtran}
\bibliography{IEEEabrv,mybib}

\begin{thebibliography}{10}
\providecommand{\url}[1]{#1}
\csname url@samestyle\endcsname
\providecommand{\newblock}{\relax}
\providecommand{\bibinfo}[2]{#2}
\providecommand{\BIBentrySTDinterwordspacing}{\spaceskip=0pt\relax}
\providecommand{\BIBentryALTinterwordstretchfactor}{4}
\providecommand{\BIBentryALTinterwordspacing}{\spaceskip=\fontdimen2\font plus
\BIBentryALTinterwordstretchfactor\fontdimen3\font minus
  \fontdimen4\font\relax}
\providecommand{\BIBforeignlanguage}[2]{{%
\expandafter\ifx\csname l@#1\endcsname\relax
\typeout{** WARNING: IEEEtran.bst: No hyphenation pattern has been}%
\typeout{** loaded for the language `#1'. Using the pattern for}%
\typeout{** the default language instead.}%
\else
\language=\csname l@#1\endcsname
\fi
#2}}
\providecommand{\BIBdecl}{\relax}
\BIBdecl

\bibitem{hu2014energy}
R.~Hu and Y.~Qian, ``An energy efficient and spectrum efficient wireless
  heterogeneous network framework for {5G} systems,'' \emph{{IEEE} Commun.
  Mag.}, vol.~52, no.~5, pp. 94--101, May 2014.

\bibitem{han2015traffic}
T.~Han and N.~Ansari, ``A traffic load balancing framework for software-defined
  radio access networks powered by hybrid energy sources,'' \emph{{IEEE/ACM}
  Trans. Netw.}, vol.~24, no.~2, pp. 1038--1051, Apr. 2016.

\bibitem{shakir2013green}
M.~Z. Shakir, K.~Qaraqe, H.~Tabassum, M.-S. Alouini, E.~Serpedin, M.~A. Imran
  \emph{et~al.}, ``Green heterogeneous small-cell networks: toward reducing the
  {CO}$_2$ emissions of mobile communications industry using uplink power
  adaptation,'' \emph{{IEEE} Commun. Mag.}, vol.~51, no.~6, pp. 52--61, Jun.
  2013.

\bibitem{han2013greening}
T.~Han and N.~Ansari, ``On greening cellular networks via multicell
  cooperation,'' \emph{{IEEE} Wireless Commun. Mag.}, vol.~20, no.~1, pp.
  82--89, Feb. 2013.

\bibitem{ender15}
K.~Davaslioglu and E.~Ayanoglu, ``Quantifying potential energy efficiency gain
  in green cellular wireless networks,'' \emph{{IEEE} Commun. Surveys Tuts.},
  vol.~16, no.~4, pp. 2065--2091, Nov. 2014.

\bibitem{niu2010cell}
Z.~Niu, Y.~Wu, J.~Gong, and Z.~Yang, ``Cell zooming for cost-efficient green
  cellular networks,'' \emph{{IEEE} Commun. Mag.}, vol.~48, no.~11, pp. 74--79,
  Nov. 2010.

\bibitem{wu2012green}
J.~Wu, S.~Rangan, and H.~Zhang, \emph{Green Communications: Theoretical
  Fundamentals, Algorithms and Applications}.\hskip 1em plus 0.5em minus
  0.4em\relax CRC Press, 2012.

\bibitem{li2011energy}
G.~Y. Li, Z.~Xu, C.~Xiong, C.~Yang, S.~Zhang, Y.~Chen, and S.~Xu,
  ``Energy-efficient wireless communications: tutorial, survey, and open
  issues,'' \emph{{IEEE} Wireless Commun. Mag.}, vol.~18, no.~6, pp. 28--35,
  Dec. 2011.

\bibitem{cui2004energy}
S.~Cui, A.~J. Goldsmith, and A.~Bahai, ``Energy-efficiency of {MIMO} and
  cooperative {MIMO} techniques in sensor networks,'' \emph{IEEE J. Sel. Areas
  Commun.}, vol.~22, no.~6, pp. 1089--1098, Aug. 2004.

\bibitem{wangxin2013}
X.~Wang and Z.~Li, ``Energy-efficient transmissions of bursty data packets with
  strict deadlines over time-varying wireless channels,'' \emph{{IEEE} Trans.
  Wireless Commun.}, vol.~12, no.~5, pp. 2533--2543, 2013.

\bibitem{qing15_wpcn_twc}
Q.~Wu, M.~Tao, D.~W.~K. Ng, W.~Chen, and R.~Schober, ``Energy-efficient
  resource allocation for wireless powered communication networks,''
  \emph{{IEEE} Trans. Wireless Commun.}, to appear, 2015.

\bibitem{li2014energy}
H.~Li, L.~Song, and M.~Debbah, ``Energy efficiency of large-scale multiple
  antenna systems with transmit antenna selection,'' \emph{{IEEE} Trans.
  Commun.}, vol.~62, no.~2, pp. 638--647, Feb. 2014.

\bibitem{ng2012energy1}
D.~W.~K. Ng, E.~S. Lo, and R.~Schober, ``Energy-efficient resource allocation
  in {OFDMA} systems with large numbers of base station antennas,''
  \emph{{IEEE} Trans. Wireless Commun.}, vol.~11, no.~9, pp. 3292--3304, Sep.
  2012.

\bibitem{ng2012energy3}
------, ``Energy-efficient resource allocation for secure {OFDMA} systems,''
  \emph{IEEE Trans. Veh. Technol.}, vol.~61, no.~6, pp. 2572--2585, Jul. 2012.

\bibitem{ng2012energy2}
------, ``Energy-efficient resource allocation in multi-cell {OFDMA} systems
  with limited backhaul capacity,'' \emph{{IEEE} Trans. Wireless Commun.},
  vol.~11, no.~10, pp. 3618--3631, Oct. 2012.

\bibitem{cui2014optimal}
Q.~Cui, X.~Yang, J.~Hamalainen, X.~Tao, and P.~Zhang, ``Optimal
  energy-efficient relay deployment for the bidirectional relay transmission
  schemes,'' \emph{IEEE Trans. Veh. Technol.}, vol.~63, no.~6, pp. 2625--2641,
  Jul. 2014.

\bibitem{sun2013energy}
C.~Sun, Y.~Cen, and C.~Yang, ``Energy efficient {OFDM} relay systems,''
  \emph{{IEEE} Trans. Commun.}, vol.~61, no.~5, pp. 1797--1809, May 2013.

\bibitem{liu15}
G.~Liu, F.~R. Yu, H.~Ji, and V.~Leung, ``Energy-efficient resource allocation
  in cellular networks with shared full-duplex relaying,'' \emph{IEEE Trans.
  Veh. Technol.}, vol.~64, no.~8, pp. 3711--3724, Aug. 2015.

\bibitem{yiran2015}
Y.~Xu, R.~Hu, Y.~Qian, and T.~Znati, ``Video quality-based spectral and energy
  efficient mobile association in heterogeneous wireless networks,''
  \emph{{IEEE} Trans. Commun.}, 2015, early access.

\bibitem{ramamonjison2015energy}
R.~Ramamonjison and V.~K. Bhargava, ``Energy efficiency maximization framework
  in cognitive downlink two-tier networks,'' \emph{{IEEE} Trans. Wireless
  Commun.}, vol.~14, no.~3, pp. 1468--1479, Mar. 2015.

\bibitem{han2014spectrum}
S.~Han, C.~Yang, and A.~F. Molisch, ``Spectrum and energy efficient cooperative
  base station doze,'' \emph{{IEEE} J. Sel. Areas Commun.}, vol.~32, no.~2, pp.
  285--296, Feb. 2014.

\bibitem{sheenergy}
C.~She, C.~Yang, and L.~Liu, ``Energy-efficient resource allocation for
  {MIMO}-{OFDM} systems serving random sources with statistical {QoS}
  requirement,'' \emph{{IEEE} Trans. Wireless Commun.}, vol.~63, no.~11, pp.
  4125--4141, Nov. 2015.

\bibitem{ge2014energy}
X.~Ge, X.~Huang, Y.~Wang, M.~Chen, Q.~Li, T.~Han, and C.~Wang,
  ``Energy-efficiency optimization for {MIMO}-{OFDM} mobile multimedia
  communication systems with {QoS} constraints,'' \emph{IEEE Trans. Veh.
  Technol.}, vol.~63, no.~5, pp. 2127--2138, Jun. 2014.

\bibitem{soh2013energy}
Y.~S. Soh, T.~Q. Quek, M.~Kountouris, and H.~Shin, ``Energy efficient
  heterogeneous cellular networks,'' \emph{{IEEE} J. Sel. Areas Commun.},
  vol.~31, no.~5, pp. 840--850, May 2013.

\bibitem{xie2012energy}
R.~Xie, F.~R. Yu, H.~Ji, and Y.~Li, ``Energy-efficient resource allocation for
  heterogeneous cognitive radio networks with femtocells,'' \emph{{IEEE} Trans.
  Wireless Commun.}, vol.~11, no.~11, pp. 3910--3920, Nov. 2012.

\bibitem{guo2014joint}
Y.~Guo, J.~Xu, L.~Duan, and R.~Zhang, ``Joint energy and spectrum cooperation
  for cellular communication systems,'' \emph{{IEEE} Trans. Commun.}, vol.~62,
  no.~10, pp. 3678--3691, Oct. 2014.

\bibitem{cao2015cognitive}
X.~Cao, Y.~Chen, and K.~Liu, ``Cognitive radio networks with heterogeneous
  users: How to procure and price the spectrum?'' \emph{{IEEE} Trans. Wireless
  Commun.}, vol.~14, no.~3, pp. 1676--1688, Mar. 2015.

\bibitem{han2014enabling}
T.~Han and N.~Ansari, ``Enabling mobile traffic offloading via energy spectrum
  trading,'' \emph{{IEEE} Trans. Wireless Commun.}, vol.~13, no.~6, pp.
  3317--3328, Jun. 2014.

\bibitem{tse2005fundamentals}
D.~Tse and P.~Viswanath, \emph{Fundamentals of wireless communication}.\hskip
  1em plus 0.5em minus 0.4em\relax Cambridge university press, 2005.

\bibitem{yumulti15}
G.~Yu, Y.~Jiang, L.~Xu, and G.~Li, ``Multi-objective energy-efficient resource
  allocation for multi-{RAT} heterogeneous networks,'' \emph{IEEE J. Sel. Areas
  Commun.}, vol.~33, no.~10, pp. 2118 -- 2127, Oct. 2015.

\bibitem{ismailsurvey14}
M.~Ismail, W.~Zhuang, E.~Serpedin, and K.~Qaraqe, ``A survey on green mobile
  networking: From the perspectives of network operators and mobile users,''
  \emph{{IEEE} Commun. Surveys Tuts.}, vol.~49, no.~6, pp. 30--37, Dec. 2014.

\bibitem{auer2010d2}
G.~Auer, O.~Blume, V.~Giannini, I.~Godor, M.~Imran, Y.~Jading, E.~Katranaras,
  M.~Olsson, D.~Sabella, P.~Skillermark \emph{et~al.}, ``D2. 3: Energy
  efficiency analysis of the reference systems, areas of improvements and
  target breakdown,'' \emph{EARTH NFSO-ICT247733}, 2010. [Online]. Available:
  http://cordis.europa.eu/docs/projects/cnect/3/247733/080/deliverables/001-EARTHWP2D23v2.pdf.

\bibitem{derrick_harvest13}
D.~W.~K. Ng, E.~S. Lo, and R.~Schober, ``Energy-efficient resource allocation
  in {OFDMA} systems with hybrid energy harvesting base station,'' \emph{{IEEE}
  Trans. Wireless Commun.}, vol.~12, no.~7, pp. 3412--3427, Jul. 2013.

\bibitem{dinkelbach1967nonlinear}
W.~Dinkelbach, ``On nonlinear fractional programming,'' \emph{Management
  Science}, vol.~13, no.~7, pp. 492--498, Mar. 1967.

\bibitem{Boyd}
S.~Boyd and L.~Vandenberghe, \emph{Convex Optimization}.\hskip 1em plus 0.5em
  minus 0.4em\relax Cambridge University Press, 2004.

\bibitem{corless1996lambertw}
R.~M. Corless, G.~H. Gonnet, D.~E. Hare, D.~J. Jeffrey, and D.~E. Knuth, ``On
  the lambertw function,'' \emph{Advances in Computational mathematics},
  vol.~5, no.~1, pp. 329--359, Dec. 1996.

\bibitem{yu2006dual}
W.~Yu and R.~Lui, ``Dual methods for nonconvex spectrum optimization of
  multicarrier systems,'' \emph{{IEEE} Trans. Commun.}, vol.~54, no.~7, pp.
  1310--1322, Jul. 2006.

\bibitem{bertsekas1999nonlinear}
D.~P. Bertsekas, ``Nonlinear programming,'' 1999.

\bibitem{miao2010energy}
G.~Miao, N.~Himayat, and G.~Y. Li, ``Energy-efficient link adaptation in
  frequency-selective channels,'' \emph{{IEEE} Trans. Commun.}, vol.~58, no.~2,
  pp. 545--554, Feb. 2010.

\bibitem{loodaricheh2014energy}
R.~A. Loodaricheh, S.~Mallick, and V.~K. Bhargava, ``Energy-efficient resource
  allocation for {OFDMA} cellular networks with user cooperation and {QoS}
  provisioning,'' \emph{{IEEE} Trans. Wireless Commun.}, vol.~13, no.~11, pp.
  6132--6146, Nov 2014.

\bibitem{cheung2012achieving}
K.~Cheung, S.~Yang, and L.~Hanzo, ``Achieving maximum energy-efficiency in
  multi-relay {OFDMA} cellular networks: A fractional programming approach,''
  \emph{{IEEE} Trans. Commun.}, vol.~61, no.~7, pp. 2746--2757, Jul. 2013.

\bibitem{cormen2009introduction}
T.~H. Cormen, \emph{Introduction to algorithms}.\hskip 1em plus 0.5em minus
  0.4em\relax MIT press, 2009.

\bibitem{ma2016resource}
X.~Ma, J.~Liu, and H.~Jiang, ``Resource allocation for heterogeneous
  applications with device-to-device communication underlaying cellular
  networks,'' \emph{{IEEE} J. Sel. Areas Commun.}, vol.~34, no.~1, pp. 15--26,
  2016.

\bibitem{ngo2014joint}
D.~T. Ngo, S.~Khakurel, and T.~Le-Ngoc, ``Joint subchannel assignment and power
  allocation for {OFDMA} femtocell networks,'' \emph{{IEEE} Trans. Wireless
  Commun.}, vol.~13, no.~1, pp. 342--355, Jan. 2014.

\bibitem{saeed2013energy}
A.~Saeed, A.~Akbari, M.~Dianati, and M.~A. Imran, ``Energy efficiency analysis
  for {LTE} macro-femto hetnets,'' in \emph{Proc. VDE EW}, 2013, pp. 1--5.

\bibitem{liu2013massive}
W.~Liu, S.~Han, C.~Yang, and C.~Sun, ``Massive {MIMO} or small cell network:
  Who is more energy efficient?'' in \emph{Proc. IEEE WCNC}, 2013, pp. 24--29.

\end{thebibliography}

\end{document}